\documentclass[11pt]{article}

\title{The low energy dynamics of charge two dyonic instantons}
\author{James P Allen\footnote{Email: j.p.allen@durham.ac.uk}\; and Douglas J Smith\footnote{Email: douglas.smith@durham.ac.uk}}

\usepackage[intlimits]{mathtools}
\usepackage{amsthm} 
\usepackage{amssymb}
\usepackage{commath}
\usepackage{bbold}
\usepackage[textwidth=430pt, textheight=630pt]{geometry}
\usepackage{hyperref}
\usepackage{color}
\usepackage{graphicx}
\usepackage{caption}
\usepackage{subcaption}
\usepackage{fancyhdr}
\usepackage{setspace}

\newtheorem{claim}{Claim}[section]

\newcommand{\id}{\mathbb{1}}
\newcommand{\R}{\mathbb{R}}
\newcommand{\C}{\mathbb{C}}
\newcommand{\SU}{\mathrm{SU}}
\renewcommand{\O}{\mathrm{O}}
\newcommand{\U}{\mathrm{U}}
\DeclareMathOperator{\Tr}{Tr}
\renewcommand{\Re}{\mathop{\mathrm{Re}}}
\renewcommand{\Im}{\mathop{\mathrm{Im}}}

\newcommand{\T}{\intercal}
\DeclareMathOperator{\spn}{span}

\numberwithin{equation}{section}

\makeatletter
\@addtoreset{equation}{section}
\makeatother

\begin{document}
\maketitle
\thispagestyle{fancy}
\rhead{\large \bf DCPT-12/39}
\begin{center}
\emph{Department of Mathematical Sciences, Durham University,\\
Lower Mountjoy, Stockton Road, Durham DH1 3LE, UK}
\end{center}
\vspace{1.0cm}

\begin{abstract}
We explore the low energy dynamics of charge two instantons and dyonic instantons in $\SU(2)$ $5$-dimensional Yang-Mills. We make use of the moduli space approximation and first calculate the moduli space metric for two instantons. For dyonic instantons the effective action of the moduli space approximation also includes a potential term which we calculate. Using the ADHM construction we are able to understand some aspects of the topology and structure of the moduli space. We find that instantons undergo right angled scattering after a head on collision and we are able to give an analytic description of this in terms of a quotient of the moduli space by symmetries of the ADHM data. We also explore the scattering of instantons and dyonic instantons numerically in a constrained region of the moduli space. Finally we exhibit some examples of closed geodesics on the moduli space, and geodesics which hit the moduli space singularities in finite time.
\end{abstract}

\section{Introduction}

Instantons were first studied as topological solitons in Yang-Mills theories over $45$ years ago \cite{Belavin197585}. They have since played an important role in our understanding of non-pertubative phenomena in supersymmetric Yang-Mills. Detailed reviews can be found in \cite{Dorey:2002ik, Shifman:1999mv, Tong:2005un}. However, much less is known about the dynamics of instantons compared to other soliton systems such as their lower dimensional cousins, monopoles.

Instantons were first studied in $4$-dimensional Euclidean Yang-Mills where there are no questions to be asked about dynamics. However, instantons can also be embedded in higher dimensional Yang-Mills theories where there is an additional time component. In particular, instantons are interesting objects on D$4$-branes where the low energy world volume theory is $5$ dimensional super-Yang-Mills. Here instantons are $1/2$ BPS states corresponding to D$0$-branes dissolved in the D$4$-branes \cite{Douglas:1995bn, Witten:1995gx}. The $4$ dimensional Euclidean instanton solutions lift to the $4+1$ dimensional world volume as static solutions. Instantons on D$4$-branes are also significant when viewed in the compactification limit of M$5$-branes. In this case instantons carry KK momentum along the circle \cite{Aharony:1997th, Aharony:1997an}. In addition, it has recently been conjectured that the low energy theory of D$4$-branes is dual to the M$5$-brane theory on $\R^{4,1} \times S^1$ and is UV complete when including instantons and monopoles \cite{Lambert:2010iw, Douglas:2010iu}. If this is correct then non-perturbative states, including instantons, must carry information about the higher dimensional theory of M$5$-branes.

Understanding the dynamics of instantons in the full field theory is generically too complicated. Instead, one can use the powerful moduli space approximation of Manton \cite{Manton1982} to study the motion of instantons as geodesics on the instanton moduli space. The space of instanton solutions of a given topological charge, $k$, make up the moduli space and this can be parameterised by a finite number of parameters. With an $\SU(N)$ gauge group the moduli space of instantons has dimension of $4kN$. The moduli space contains the minimum energy solutions in Yang-Mills and energy conservation requires that after giving these static minimum energy solutions a small velocity, they must remain close to the moduli space despite now having an energy slightly above the minimum. This motion can then be approximated by geodesic motion actually on the moduli space, where at each snapshot in time, the fields look like a static instanton configuration.

The world volume theory of D$4$-branes also contains six scalar fields corresponding to the branes' transverse separation. Separating the branes along one direction gives a single scalar field a non-zero expectation value. The minimum energy solutions now have a non-zero scalar field but the instanton BPS equations for the spatial components of the gauge field remain unchanged. The minimum energy solutions are now $1/4$ BPS and the $A_0$ component is proportional to the Higgs field \cite{Lambert:1999ua}. Like with monopole dyons, this can be viewed as a constant rotation in the gauge group and gives the instantons an electric charge. The electric charge exactly balances the attractive force that the non-zero Higgs field creates. The presence of a non-zero Higgs field introduces a potential on the moduli space \cite{Lambert:1999ua}. This would cause the instantons to collapse if they were not stabilised by a rotation in gauge group, giving them electric charge and balancing the attractive scalar force. From the string theory point of view, these dyonic instantons are a bound state of fundamental strings and D$0$-branes between the D$4$-branes \cite{Lambert:1999ua, Zamaklar:2000tc} and can be interpreted as supertubes \cite{Kim:2003gj, Choi:2007mk}.

The moduli space of instantons has singularities corresponding to instantons of zero size.  These singularities correspond to the transition between the Coulomb branch where the D$0$-branes are dissolved in the D$4$-branes, and the Higgs branch where they are separated. However, if we restrict ourselves solely to the world volume theory then these singularities can be regulated by turning on non-commutativity \cite{Nekrasov:1998ss}. This places a bound on the minimum size of the instanton. However, if the conjecture that $5$D super-Yangs-Mills is UV complete is true then this should be unnecessary, at least in the full QFT.

The motion of dyonic instantons has been previously studied for a single instanton and two well separated instantons \cite{Peeters:2001np}. In the presence of a potential a single slow moving dyonic instanton will sit in a stable orbit in the gauge group at a stable size. There is a conserved angular momentum which prevents this orbit from decaying and the instanton will always have a finite size \cite{Lambert:1999ua, Peeters:2001np}. This protects the dyonic instanton from evolving to the singularity. For multiple instantons only the total angular momentum is conserved and it is not clear \emph{a priori} whether the system may evolve from non-singular initial conditions to a singularity in finite time. 

The outline of the rest of this paper is as follows. In Section 2 we review dyonic instantons in $5$ dimensional Yang-Mills. We describe the ADHM construction for finding solutions to the self-dual field equation and for the scalar field in the background of an instanton. We review the construction of the moduli space metric as the inner product of zero-mode perturbations of the instantons and we show how to construct the effective action for the approximate motion of slow moving instantons.

In Section 3 we explicitly construct the metric and potential on the moduli space of charge two instantons. The metric and potential have been partially constructed before for well separated instantons \cite{Peeters:2001np} but to explore the close range behaviour it is necessary to extend this calculation to the full moduli space. Our method for calculating the metric follows that of Osborn \cite{Osborn:1981yf} and Peeters and Zamaklar \cite{Peeters:2001np}. The potential is constructed as for monopole dyons as the square of the Killing vector corresponding to global gauge rotations in the direction of the scalar expectation value \cite{Tong:1999mg, Bak:1999da, Bak:1999sv}. 

In Section 4 we briefly review the dynamics of both a single instanton and a dyonic instanton to help put the results in later sections in context.

In Sections 5 and 6 we use the metric and potential to explore the scattering of two instantons and dyonic instantons. We see that instantons in a head on collision undergo right angled scattering and are able to understand this analytically through symmetries in the ADHM data. To further explore the range of possible scattering behaviours we numerically solve the equations for geodesic motion on the moduli space for a range of initial conditions.

In Section 7 we investigate the geodesic completeness of the moduli space. Instantons are susceptible to slow roll perturbations in their size and typically either spread out indefinitely or hit the zero-size singularity after a small initial perturbation away from static. For a single dyonic instanton there is a non-zero conserved angular momentum which protects the dyonic instanton from shrinking to zero size or spreading out. However. for multiple dyonic instantons this angular momentum can be transferred between them and we see that it is again possible for them to shrink to zero size and hit the singularity on the moduli space.

Finally in Section 8 we look at the stability of a localised charge two instanton. Again, this configuration for instantons is unstable but the presence of a potential on the moduli space for dyonic instantons serves to stabilise the charge two dyonic instanton. We also find a closed geodesic on the moduli space corresponding to a rotating charge two dyonic instanton.

\section{Dyonic instantons}

We begin in Section \ref{sec:instantons in Yang-Mills} with a short review of $5$ dimensional Yang-Mills. This theory contains both instantons, which are $1/2$-BPS solutions of minimal energy in each topological sector, and dyonic instantons which are the corresponding $1/4$ BPS minimal energy solutions in the presence of a non-zero scalar expectation value. In Section \ref{sec:adhm construction} we review the ADHM construction which reduces the problem of finding instantons to an algebraic constraint. The ADHM construction also provides a natural way to parameterise the instanton solutions. In Section \ref{sec:moduli space metric} we introduce the moduli space and see how we can approximate slow moving instantons by geodesic motion on this moduli space. In Section \ref{sec:moduli space potential} we see how the non-zero scalar expectation value of dyonic instantons introduces a potential term into the effective system describing the low speed motion.

\subsection{Instantons in Yang-Mills} \label{sec:instantons in Yang-Mills}

A stack of D$4$-branes preserves half of the supersymmetries of the full space-time and is therefore described by an $\mathcal{N} = 2$ supersymmetric theory in $5$ dimensions. Open strings between different D-branes induce a $\U(N)$ world volume gauge symmetry with a vector potential, $A_\mu$, $\mu = 0,1, \ldots, 4$. The theory has five scalars, $X^I$, $I = 5, \ldots 9$, which describe the branes' positions in the transverse space. T-duality interchanges $A_\mu$ and $X^I$ so the scalars are also in the adjoint representation of the gauge group. We can factor out the center of mass from the theory so that the low energy dynamics of a stack of $N$ D$4$-branes is described by $5$ dimensional $\SU(N)$ super-Yang-Mills.

Instantons are $1/2$-BPS objects preserving half of the supersymmetry and dyonic instantons are $1/4$-BPS objects. However, in the rest of this paper the fermionic fields will have little effect so we will only consider the bosonic fields. The bosonic action is
\begin{equation}
S = - \int \dif{}^5 x \, \Tr \del{ \tfrac{1}{4} F_{\mu \nu} F^{\mu \nu} + \tfrac{1}{2} D_\mu X^I D^\mu X^I }.
\end{equation}
We have set the Yang-Mills coupling to one since it will not play a role in our calculations.
The fields are taken to be Hermetian so that the covariant derivative is given by
\begin{equation}
D_\mu X^I = \partial_\mu X^I - i [A_\mu, X^I],
\end{equation}
and the field strength is
\begin{equation}
F_{\mu \nu} = \partial_\mu A_\nu - \partial_\nu A_\mu - i [A_\mu, A_\nu].
\end{equation}
We will only consider a stack of two D$4$-branes so that our gauge group is $\SU(2)$. Working in higher gauge groups would make the explicit construction of instantons and the moduli space much harder.

We use the indices $\mu = 0, \ldots, 4$ to denote the world-volume directions of the D$4$-branes and the indices $i=1, \ldots, 4$ to denote just the spatial directions. The indices $I = 5, \ldots 9$ denote the directions transverse to the branes.

From now on we will assume that only one of the scalar fields is non-zero, $\phi \equiv X^5$. The vacuum expectation value (VEV) of $\phi$ corresponds to the separation of the D$4$-branes in the $X^5$ direction.

The energy of the system is
\begin{equation}
E = \int \dif{}^4x \, \Tr \del{ \tfrac{1}{2} F_{i0} F_{i0} + \tfrac{1}{4} F_{ij} F_{ij} + \tfrac{1}{2} D_0 \phi D_0 \phi + \tfrac{1}{2} D_i \phi D_i \phi},
\end{equation}
and the topological charge, $k$, and electric charge, $Q_E$, are respectively,
\begin{align}
k &= - \frac{1}{16 \pi^2} \int \dif{}^4x \, \varepsilon_{ijkl} \Tr \del{ F_{ij} F_{kl} }, \\
Q_E &= \int \dif{}^4x \, \Tr \del{ D_i \phi F_{i0} } = \int \dif{}^4x \, \Tr \del{ D_i \phi }^2.
\end{align}
By completing the square of the energy density in the usual Bogomolny bound way we see that the energy is bounded by
\begin{equation}
E \ge 2\pi^2 |k| + |Q_E|.
\end{equation}
For a given topological charge and electric charge this energy bound is saturated when
\begin{align}
F_{ij} &= \tfrac{1}{2}\varepsilon_{ijkl} F_{kl}, \\
F_{i0} &= D_i \phi, \\
D_0 \phi &= 0.
\end{align}
These are the BPS equations for dyonic instantons. The second and third equation are satisfied when the fields are static and $A_0 = \phi$. However, the equation of motion for the scalar field,
\begin{equation}
D^2 \phi = 0,
\end{equation}
still needs to be satisfied in the background of the self-dual gauge field. For each possible scalar VEV the equation of motion for $\phi$ has a unique solution in the background of a given gauge field. When the VEV is zero the solution for $\phi$ is everywhere zero and the remaining BPS equation is that of a pure instanton which has a self-dual field strength. With a non-zero VEV, only $A_0$ and $\phi$ are modified compared to the zero VEV solution. Therefore dyonic instantons of a given VEV are in a one-to-one correspondence with instantons and have identical self-dual spatial components of the gauge field, $A_i$.

In the next section we will see how to use the ADHM construction to solve the BPS equations and construct a self dual gauge field and scalar field, $\phi$, satisfying its equation of motion in the background of this gauge field. 

\subsection{The ADHM construction} \label{sec:adhm construction}

The ADHM construction provides a method for finding (dyonic) instanton solutions without having to solve the self-dual field equation and equation of motion for $\phi$ directly. The problem is recast from a set of differential field equations into a set of algebraic constraints that are much easier to solve in practice.

We will only consider the ADHM construction for an $\SU(2)$ gauge group so that all quantities can be expressed in terms of quaternions. We represent a quaternion, $p$, as $p = p^i e_i$ where
\begin{align}
e_a &= i \sigma_a, \; a = 1,2,3, \\
e_4 &= \id_2,
\end{align}
are $2 \times 2$ matrices satisfying the quaternion algebra $e_1^2 = e_2^2 = e_3^2 = e_1 e_2 e_3 = -1$. The imaginary quaternions are isomorphic to $\mathfrak{su}(2)$ as can seen from their representation as Pauli matrices above. The unit quaternions are isomorphic to $\SU(2)$ which can again easily be seen from the form of the matrices in this representation. The adjoint fields can be written as
\begin{equation}
A_\mu = i A_\mu^i e_i, \quad \phi = i \phi^i e_i, \quad i = 1,2,3,
\end{equation}
where the components $A_\mu^i$ and $\phi^i$ are real. In these conventions, the anti-symmetric tensor is given by $\varepsilon_{1234} = -1$.

The starting point in the construction of a charge $k$ instanton is a $(k + 1) \times k$ quaternion valued matrix, $\Delta(x)$, known as the ADHM data. The $x$ dependence is over the spatial coordinates of the D$4$-brane world volume. The matrix $\Delta$ must satisfy the ADHM constraint,
\begin{equation} \label{eq:ADHM constraint}
\Delta^\dagger \Delta = \id_2 \otimes f^{-1},
\end{equation}
where $f^{-1}$ is a real and invertible $k \times k$ matrix and $\id_2$ is the real quaternion in the $2 \times 2$ matrix representation.

Without loss of generality we can work with $\Delta$ in the form
\begin{equation}
\Delta(x) = a - bx, 
\end{equation}
where $a$ and $b$ are block matrices,
\begin{equation}
a = \begin{pmatrix}
\Lambda \\
\Omega
\end{pmatrix}, \quad
b = \begin{pmatrix}
0 \\
\id_k
\end{pmatrix}.
\end{equation}
The top entry, $\Lambda$, is a row vector of length $k$ and $\Omega$ is a $k \times k$ symmetric matrix. The ADHM constraint is then a constraint on $a$,
\begin{equation}
a^\dagger a = \id_2 \otimes \mu^{-1},
\end{equation}
where $\mu^{-1}$ is a real and invertible $k \times k$ matrix.

To find the gauge field we must first find a quaternionic column vector, $U(x)$, satisfying
\begin{equation} \label{eq:U definition}
\Delta^\dagger U = 0, \quad \text{and} \quad U^\dagger U = \id_2.
\end{equation}
The gauge field is then given by
\begin{equation}
A_i = i U^\dagger \partial_i U.
\end{equation}
The field strength is self-dual so long as the ADHM constraint in equation \eqref{eq:ADHM constraint} is satisfied. Note that $U$ is uniquely determined up to a transformation of the form
\begin{equation}
U \rightarrow U\Omega(x),
\end{equation}
where $\Omega(x) \in \U(N)$ is a gauge transformation of $A_i$. The ADHM construction therefore doesn't pick out any specific gauge and should be thought of as mapping $\Delta(x)$ to a set of physically equivalent gauge fields.

For a dyonic instanton the scalar field can be constructed from the ansatz \cite{Dorey:1996hu},
\begin{equation} \label{eq:phi ansatz}
\phi = i U^\dagger \begin{pmatrix}
q & 0 \\
0 & P
\end{pmatrix} U,
\end{equation}
where $q$ is a pure imaginary quaternion and $iq$ is the vacuum expectation value (VEV) of $\phi$ at infinity. The lower block, $P$, is a real $k \times k$ antisymmetric matrix. For $\phi$ to satisfy its equation of motion, $P$ must satisfy
\begin{equation}
2 \eta^a_{ij} q^a \Lambda^\T_i \Lambda_j - [\Omega^\T_i, [\Omega_i, P]] - \{P, \Lambda^\T_i \Lambda_i \} = 0.
\end{equation}
This is shown in Appendix \ref{ap:calculation of phi}.

The fields are invariant under transformations of the ADHM data of the form
\begin{equation} \label{eq:rotation of ADHM data}
\Delta \rightarrow Q \Delta R^{-1} = \begin{pmatrix}
p & 0 \\
0 & R
\end{pmatrix} \Delta R^{-1}, \quad \text{and} \quad
U \rightarrow QU,
\end{equation}
where $p$ is a unit quaternion and $R$ is an orthogonal $k \times k$ matrix. The scalar field is only invariant when the imaginary part of $p$ is parallel to $q$, and otherwise has its VEV rotated by $iq \rightarrow i \bar p q p$.

Na\"ively we count $4k$ real parameters in $\Lambda$ and $2k(k+1)$ real parameters in $\Omega$. These must satisfy the ADHM constraints which remove $\tfrac{3}{2} k(k-1)$ degrees of freedom. The transformation in equation \eqref{eq:rotation of ADHM data} has $\tfrac{1}{2} k(k-1)$ parameters in $R$ and $3$ in $p$. These are redundancies in the parameterisation of the ADHM data. Altogether there are therefore $8k - 3$ degrees of freedom in the ADHM data. We also include global gauge transformations in our counting of degrees of freedom which brings the total to $8k$. The ADHM construction is complete so that the space of solutions to the self-dual field equations with charge $k$ is $8k$ dimensional.

If there is a non zero scalar field then there is only one parameter in $p$ that leaves the expectation value invariant. However, there is only one unbroken global gauge symmetry and our counting still applies.

\subsection{The moduli space of instantons} \label{sec:moduli space metric}

The moduli space is the space of all instanton solutions of a given instanton number $k$. For any solution, all gauge transformations of that solution will also be a solution, but since these are all physically equivalent we choose to discount them from the moduli space. More precisely, the moduli space is therefore the space of all solutions to the self-dual field equations quotiented by local gauge transformations. Conventionally the moduli space still includes global gauge transformations for reasons that will become apparent later.

The moduli space has a naturally induced metric which can be derived in two equivalent ways. We can consider the geometry induced by varying from one point in the moduli space to another, or alternatively by considering the motion of slow moving instantons in the full field theory. The first is intrinsic to the moduli space while the second has a more physical interpretation.

Let us consider the intrinsic geometry on the moduli space first. We will begin by considering the unquotiented moduli space which includes local gauge transformations. Each point on this space corresponds to some instanton solution, $A_i$, and the tangent space at each point is made up of infinitesimal variations of the gauge field, $\delta A_i$, which remain in the moduli space. If $A_i \rightarrow A_i + \delta A_i$ is to remain in the moduli space it must satisfy the self-dual field equation and $\delta A_i$ must therefore satisfy the linear self-dual field equation,
\begin{equation} \label{eq:linear self-dual field equation}
D_i ( \delta A_j ) - D_j ( \delta A_i ) = \varepsilon_{ijkl} D_k ( \delta A_l ).
\end{equation}
Note that gauge transformations automatically satisfy this linear self-dual field equation.

There is a natural inner product on the unquotiented moduli space,
\begin{equation}
g(\delta A_i, \delta'A_i) = \int \dif{}^4x \Tr \del{\delta A_i \delta'A_i}.
\end{equation}
This will induce a metric on the quotiented moduli space. We are working with explicit representative fields in each gauge equivalence class, so the metric on the quotiented moduli space must be zero on tangent vectors which are purely gauge transformations so that it is well defined. Rather than modify the metric we can demand that all tangent vectors, $\delta A_i$, in the quotiented moduli space are orthogonal to gauge transformations,
\begin{equation}
g(\delta A_i, D_i \Lambda) = - \int \dif{}^4x \Tr \del{D_i (\delta A_i) \Lambda} = 0,
\end{equation}
for all $\Lambda$, or equivalently,
\begin{equation}
D_i \delta A_i = 0.
\end{equation}
Variations which satisfy the linear self-dual field equation and this gauge fixing conditions are known as \emph{zero modes} and form the tangent space at each point in the quotiented moduli space.

The parameters in the ADHM data provide a natural coordinate system on the moduli space. If we label the $8k$ parameters of the ADHM construction as $z^r$, $r = 1, \ldots, 8k$ then each choice corresponds to an instanton solution, $A_i(\mathbf{z}; \mathbf{x})$. In this coordinate system on the moduli space the canonical zero-modes are
\begin{equation} \label{eq:canonical zero-mode}
\delta_r A_i = \partial_r A_i - D_i \varepsilon_r.
\end{equation}
The first term is the canonical tangent vector in the unquotiented moduli space. The second term, $D_i \epsilon_r$, removes any gauge transformation component in $\partial_r A_i$. It is chosen so that $D_i(\delta_r A_i) = 0$.

To investigate the dynamics of instantons we need to ask what happens when we give the static instantons a small velocity. Of course these configurations will no longer strictly be instantons, but for small velocities they will remain close to the minimum energy solutions. Solving the full field theory equations for their motion would be extremely complicated. However, for small velocities the problem can be approximated by motion on the moduli space \cite{Manton1982}. Since the initial field configuration starts close to a minimum energy solution, by energy conservation the evolution of the fields must always stay close to a minimum energy solution and therefore close to solutions which lie in the moduli space. As reviewed below, this motion can be approximated by geodesic motion on the moduli space, with the metric as defined above. If the coordinates on the moduli space are labelled as $z^r$, then we allow a time dependence in $A_\mu(\mathbf{z}(t); \mathbf{x})$ only through $\mathbf{z}(t)$. This procedure is well understood for monopoles and a useful review is provided in reference \cite{Weinberg:2006rq}. We follow a similar argument for instantons.

Now that our fields have a time dependence they will not automatically satisfy the Yang-Mills equations of motion as the static fields did. For instantons without any excited scalar fields, Gauss's law becomes
\begin{equation}
D_i F_{i0} = D_i(D_i A_0 - \dot z^r \partial_r A_i) = 0.
\end{equation}
This can be solved by a perturbation of $A_0$ away from zero,
\begin{equation}
A_0 = \dot z^r \epsilon_r,
\end{equation}
where $\epsilon_r$ is chosen so that $D_i(D_i \epsilon_r - \partial_r A_i) = 0$. The electric components of the field strength can now be written as
\begin{equation}
F_{i0} = - \dot z^r \delta_r A_i,
\end{equation}
where
\begin{equation}
\delta_r A_i = \partial_r A_i - D_i \epsilon_r,
\end{equation}
are the zero-modes we defined above.

Substituting these slow moving instanton solutions into the Yang-Mills action gives an effective action for motion on the moduli space in terms of $z^r(t)$,
\begin{equation}
S = \tfrac{1}{2} \int \dif{}^5x \, \Tr \left (F_{i0} F_{i0} \right) = \tfrac{1}{2} \int \dif{}t \, g_{rs} \dot z^r \dot z^s,
\end{equation}
where
\begin{equation}
g_{rs} = \int \dif{}^4x \, \Tr \left( \delta_r A_i \delta_s A_i \right),
\end{equation}
is the metric on the moduli space. The motion of slow moving instantons is therefore described by geodesic motion on the moduli space with this metric. At each snapshot in time, the fields still look like a static instanton.

\subsection{The moduli space of dyonic instantons} \label{sec:moduli space potential}

The moduli space of dyonic instantons is the same as the moduli space of instantons. Each dyonic instanton has a unique underlying instanton so that the moduli spaces can be identified. The metric is unchanged by the non-zero scalar field but each dyonic instanton solution has an electric charge, $Q_E$, which varies across the moduli space. This provides a contribution to the energy of the field configurations and so introduces a potential on the moduli space.

In the moduli space approximation the scalar field has a time dependence through the moduli space coordinates, $\phi(\mathbf{z}(t); \mathbf{x})$, just as with the gauge field. Gauss's law is
\begin{equation}
D_i F_{i 0} + [D_0 \phi, \phi] = 0,
\end{equation}
but this is no longer satisfied by the moduli space ansatz. It can be solved approximately by perturbing $A_0$ away from its static value,
\begin{equation}
A_0 = \phi + \dot z^r \epsilon_r.
\end{equation}
The electric component of the field strength, $F_{i0}$, is then
\begin{equation}
F_{i0} = - (\dot z^r \delta_r A_i - D_i \phi). 
\end{equation}
The first term in Gauss's law, $D_i F_{i0}$, is still zero since the static equation of motion for $\phi$ is $D_i D_i \phi = 0$ and is unchanged by the additional time dependence. The second term in Gauss's law, $[D_0 \phi, \phi]$, is non-zero but is of order $\dot z^r |q|^2$ where $|q|$ is the magnitude of the VEV of $\phi$. As we will discuss below, this constrains the regime in which the moduli space approximation is valid.

Before constructing the potential for the effective action on the moduli space, we note that $D_i \phi$ satisfies the same conditions as $\delta_r A_i$ for being a zero mode. It is a solution to the linear self-dual equation in equation \eqref{eq:linear self-dual field equation} and satisfies the gauge fixing condition, $D_i D_i \phi = 0$. Since the zero-modes $\delta_r A_i$, $r = 1, \ldots 8k$ form a basis for zero-modes on the moduli space, we can express $D_i \phi$ as
\begin{equation}
D_i \phi = |q| K^r \delta_r A_i,
\end{equation}
for some vector $K^r$. Note that we have factored out the magnitude of the scalar VEV. If we consider $D_i \phi$ at infinity then it is a global gauge transformation by $q$. For $\SU(2)$ there will be three zero-modes corresponding to a global gauge transformation and since a gauge rotation is a symmetry of the full Yang-Mills theory, these transformations descend to Killing vectors on the moduli space. The vector $K^r$ is therefore a Killing vector of the metric, corresponding to a global gauge transformation in the direction of the VEV of $\phi$.

As for monopoles \cite{Bak:1999sv, Bak:1999da}, we can perform a coordinate transformation of the moduli space coordinates to write the electric components of the field strength as 
\begin{equation}
F_{i0} = -(\dot z^r - |q| K^r) \delta_r A_i = -\dot y^r \delta_r A_i,
\end{equation}
where
\begin{equation} \label{eq:moduli space coordinate transformation}
y^r = z^r - |q| K^r t.
\end{equation}
In these coordinates the effective action is 
\begin{align}
S &= \tfrac{1}{2} \int \dif{}^5x \, \Tr \left (F_{i0} F_{i0} - D_i \phi D_i \phi + D_0 \phi D_0 \phi \right)\\
& = \tfrac{1}{2} \int \dif t \, g_{rs} \dot y^r \dot y^s - |q|^2 g_{rs} K^r K^s.
\end{align}

We have neglected terms of order $\dot z^2 |q|^2$. This effective action is therefore a valid approximation to the low energy dynamics of dyonic instantons when
\begin{equation}
\dot z^2 \ll 1, \quad \text{and} \quad |q|^2 \ll 1,
\end{equation}
in comparison to the rest mass of the dyonic instantons. Physically, this is the requirement that the potential on the moduli space is shallow compared to the potential around the moduli space and that the kinetic energy is small. This prevents the motion from being able to climb up the sides of the potential surrounding the moduli space and move away from the regime in which the approximation is valid.

If we have a coordinate, $\theta$, which corresponds to a rotation in the unbroken $U(1)$ global gauge then the Killing vector is
\begin{equation}
K = \frac{\partial}{\partial \theta}.
\end{equation}
The $1/4$-BPS dyonic instanton solutions, with $\dot z(t) = 0$ in the old coordinate system, are now rotating through the global gauge at a constant velocity in the new coordinate system, $\dot \theta = |q|$. More generally they satisfy
\begin{equation}
\dot y^r = |q| K^r,
\end{equation}
which matches the original description of dyonic instantons \cite{Lambert:1999ua}. From now on will use $z^r$ to denote these new coordinates.

Note that the potential on the moduli space is expressed as the square of a Killing vector of the moduli space metric,
\begin{equation}
V = \tfrac{1}{2} \int \dif{}^5x \, \Tr \left (D_i \phi D_i \phi\right) =  \tfrac{1}{2} \int \dif t \, |q|^2 g_{rs} K^r K^s.
\end{equation}
If we were to consider the full supersymmetric Yang-Mills theory then we would also have fermionic zero modes and a supersymmetric effective action on the moduli space. This form of the potential would then be required by supersymmetry.

\section{The moduli space of two dyonic instantons} \label{sec:moduli space of two dyonic instantons}

So far we have given an overview of the moduli space of $\SU(2)$ (dyonic) instantons of general charge $k$. To proceed with our study of charge two dyonic instantons we need to explicitly find the metric and potential on the moduli space. The moduli space will be $16$ dimensional and the metric and potential will be expressed in terms of $16$ coordinates coming from the parameters in the ADHM construction. Four of these coordinates can be factored out as an uninteresting center of mass and the remaining $12$ are grouped as three quaternionic parameters.

To calculate the metric we follow the method of Osborn \cite{Osborn:1981yf} and Peeters and Zamaklar \cite{Peeters:2001np}. Thanks to the ADHM construction, the metric can be calculated from the ADHM data, $\Delta$, without having to explicitly work out the zero-modes, $\delta_r A_i$, and their inner products. Recall that the variation of the gauge field in the direction of one of the coordinates, $z^r$, corresponds to a tangent vector in the moduli space. The derivative, $\partial_r A_i$, of the gauge field with respect to $z^r$ will generally include a gauge transformation which must be projected out before calculating the metric, as in equation \eqref{eq:canonical zero-mode}. This projection can instead be performed directly in the ADHM data by a transformation of the form $\Delta \rightarrow Q\Delta R$. Finding the inner product is reduced to an algebraic problem that is tractable.

The most general ADHM data for a charge two instanton can be written in the form
\begin{equation} \label{eq:ADHM data}
\Delta(x) = \begin{pmatrix}
v_1 & v_2 \\
\tilde \rho + \tau & \sigma \\
\sigma & \tilde \rho - \tau
\end{pmatrix}
- x \begin{pmatrix}
0 & 0 \\
1 & 0 \\
0 & 1
\end{pmatrix}.
\end{equation}
This form has been chosen so that the parameters all have a direct physical interpretation. The two lower diagonal entries, $\tilde \rho + \tau$ and $\tilde \rho - \tau$, can be interpreted as four-vectors rather than quaternions and describe the positions of the two instantons on the D$4$-brane. The parameter $\tilde \rho$ is the centre of mass and will factor out into a flat direction in the metric so we will set it to zero for the rest of this paper. When $\tau$ is much larger than $|v_1|$ and $|v_2|$ the instantons are well separated and form two distinct lumps. Each lump can be approximated by a charge one instanton which is rotationally symmetric. A cross section of the topological charge and scalar field is shown in Figure \ref{fig:separated instantons}. As $\tau$ decreases the individual lumps come closer together and begin to deform into each other. When the magnitude of $\tau$ is equal to the magnitude of $\sigma$, the instantons are coincident and form a single lump at the origin with axial symmetry. A cross section of the topological charge and scalar field for coincident charge two instantons is shown in Figure \ref{fig:coincident instantons}. We will discuss the role of $\sigma$ and the behaviour of coincident instantons more in Section \ref{sec:right angled scattering}. The form of the scalar field for charge two and higher dyonic instantons has been studied in detail in \cite{Kim:2003gj, Choi:2007mk}. The zeroes of the scalar field correspond to where the D$4$-branes intersect and these form a circle for coincident instantons. As the instantons separate the circle of zeroes pinches off into two loops which shrink down to be point like. This has the interpretation of supertubes between the D$4$-branes which collapse as the instantons become well separated.

\begin{figure}
  \begin{subfigure}[b]{0.47\textwidth}
    \centering
    \includegraphics[width=\textwidth]{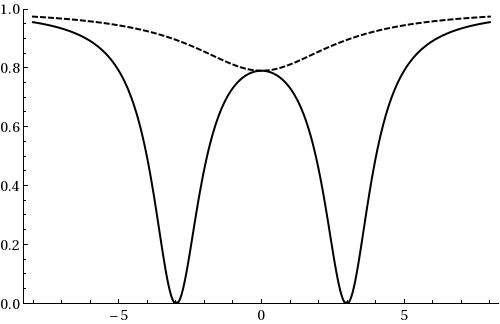}
    \caption{The scalar field.}
    \end{subfigure}
  \hfill
  \begin{subfigure}[b]{0.47\textwidth}
    \centering
    \includegraphics[width=\textwidth]{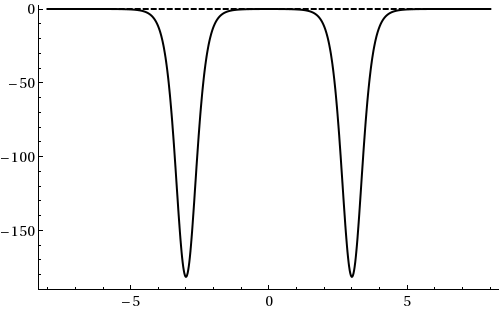}
    \caption{The topological charge density.}
  \end{subfigure}
  \caption{Two separated instantons. The solid lines show the fields along the $x_4$ (real) axis and the dashed lines show the fields along the $x_2$-axis. This configuration corresponds to the ADHM parameters $v_1 = e_4$, $v_2 = e_1$ and $\tau = 3 e_4$. Note that $|\sigma| = \tfrac{1}{6}$.}
  \label{fig:separated instantons}
\end{figure}

\begin{figure}
  \begin{subfigure}[b]{0.47\textwidth}
    \centering
    \includegraphics[width=\textwidth]{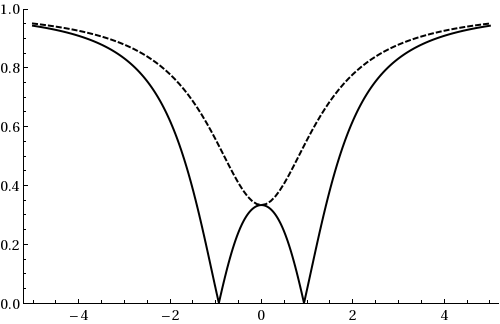}
    \caption{The scalar field.}
  \end{subfigure}
  \hfill
  \begin{subfigure}[b]{0.47\textwidth}
    \centering
    \includegraphics[width=\textwidth]{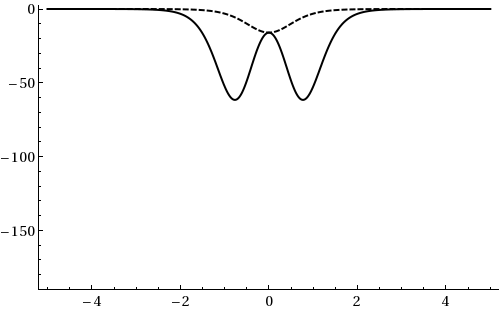}
    \caption{The topological charge.}
  \end{subfigure}
  \caption{Two coincident instantons. The solid lines show the fields along the $x_4$ (real) axis and the dashed lines show the fields along the $x_2$-axis. Note that the fields are rotationally symmetric in the $(x_4, x_1)$ plane when coincident. This configuration corresponds to the ADHM parameters $v_1 = e_4$, $v_2 = e_1$ and $\tau = \tfrac{1}{\sqrt{2}} e_4$. Note that $|\sigma| = |\tau|$.}
  \label{fig:coincident instantons}
\end{figure}

The magnitudes of $v_1$ and $v_2$ describe the size of each instanton. The unit quaternions, $\hat v_1$ and $\hat v_2$ describe their embedding in the gauge group. For well separated instantons the gauge field is approximately that of two superimposed 't Hooft instantons, but with a different global gauge transformation applied to each one,
\begin{equation}
A_i \approx \frac{|v_1|^2(x-\tau)_j \eta^a_{ij} }{|x-\tau|^2 (|x-\tau|^2 + |v_1|^2)}\, \hat v_1 \sigma^a \hat v_1^\dagger +  \frac{|v_2|^2(x+\tau)_j \eta^a_{ij}}{|x+\tau|^2 (|x+\tau|^2 + |v_2|^2)}\, \hat v_2 \sigma^a \hat v_2^\dagger.
\end{equation}
Recall that $\SU(2)$ is isomorphic to the unit quaternions so the action of $\hat v_1$ and $\hat v_2$ is just that of a gauge transformation on each separate lump. We could use the global gauge symmetry to set one of $\hat v_1$ or $\hat v_2$ to the identity matrix but we will still keep these explicit since the global gauge rotation is a relevant parameter in the moduli space and plays an important role in the dynamics. The relative gauge alignment, $\hat v_1^\dagger v_2$, of well separated instantons is physically significant even in the static case. Of course, when the instantons are close together the gauge field is more complicated and there is no clear notion of separate lumps or of the relative gauge alignment between them.

The parameter $\sigma$ is fixed by the ADHM constraint to be
\begin{equation} \label{eq:sigma definition}
\sigma = \frac{\tau}{4 |\tau|^2} \Lambda + \alpha \tau, \quad \text{where} \quad \Lambda = (\bar v_2 v_1 - \bar v_1 v_2),
\end{equation}
for some real value of $\alpha$. The symmetry of the ADHM data in equation \eqref{eq:rotation of ADHM data} always allows us to set $\alpha$ to zero. In doing so we break the continuous $\O(2)$ symmetry in equation \eqref{eq:rotation of ADHM data} to a discrete subgroup.

The calculation of the metric is quite long so to maintain the flow of this paper we will present only the final result here with the full calculation relegated to Appendix \ref{ap:calculation of metric}. In terms of the parameters $\tau$, $v_1$, and $v_2$, the metric on the moduli space is
\begin{equation} \label{eq:full metric}
\begin{split}
\frac{\dif s^2}{8 \pi^2}
&= \dif v_1^2 + \dif v_2^2 + \dif \tau^2 \\
&\qquad + \frac{1}{4|\tau|^2} \Big(|v_1|^2 \dif v_2^2 + |v_2|^2 \dif v_1^2 + 
    2 (v_1 \cdot \dif v_1) (v_2 \cdot \dif v_2) - (v_1 \cdot \dif v_2)^2 \\
&\hspace{0.6in} - (v_2 \cdot \dif v_1)^2 - 2 (v_1 \cdot v_2) (\dif v_1 \cdot \dif v_2) + 2 \varepsilon_{ijkl} v_1^i v_2^j \dif v_1^k \dif v_2^l \Big) \\
&\qquad + \frac{1}{4|\tau|^4} \left(|v_1|^2 |v_2|^2 - (v_1 \cdot v_2)^2 \right) \dif\tau^2 \\
&\qquad - \frac{1}{2|\tau|^4} \Big(|v_1|^2 (v_2 \cdot \dif v_2) + |v_2|^2 (v_1 \cdot \dif v_1)\\ 
&\hspace{0.6in} - (v_1 \cdot v_2)(v_1 \cdot \dif v_2) - (v_1 \cdot v_2)(v_2 \cdot \dif v_1) \Big) \tau \cdot \dif \tau \\
&\qquad +\frac{1}{8|\tau|^4} \Big( \varepsilon_{ijkl} \Lambda_i \dif \Lambda_j \tau_k \dif \tau_l +(\Lambda \cdot \dif \tau)(\tau \cdot \dif \Lambda) - (\Lambda \cdot \tau) (\dif \Lambda \cdot \dif \tau)  \Big) \tau \cdot \dif \tau \\
&\qquad - \frac{1}{N_A} \Big(v_1 \cdot \dif v_2 - v_2 \cdot \dif v_1 \\
&\hspace{0.6in} - \frac{2}{|\tau|^2} \big(
       \varepsilon_{mnpq} v_2^m v_1^n \tau^p \dif\tau^q + (v_2 \cdot \tau) (v_1 \cdot \dif\tau) - (v_1 \cdot \tau) (v_2 \cdot \dif\tau)
     \big)
    \Big)^2,
\end{split}
\end{equation}
where
\begin{equation}
N_A = |v_1|^2 + |v_2|^2 + 4|\tau|^2 + \frac{1}{4|\tau|^2}|\bar v_2 v_1 - \bar v_1 v_2|^2.
\end{equation}
The terms on the first line of the metric correspond to the individual movement of each instanton lump. The remaining terms describe the interaction of the two instanton lumps and the behaviour when it is not possible to distinguish two separate lumps. Note that these terms fall off quadratically as the separation is increased. This metric has been previously calculated up to order $|\tau|^{-2}$ in reference \cite{Peeters:2001np} although we point out that our calculation differs by a factor of two in the final line.

The potential on the moduli space can be calculated directly from the ansatz for $\phi$ in equation \eqref{eq:phi ansatz}. Again we will present only the result here and leave the calculation for Appendix \ref{ap:calculation of potential}. The potential is
\begin{equation}
V = 8 \pi^2 |q|^2 \del{ |v_1|^2 + |v_2|^2 - \frac{1}{N_A}|\bar v_2 \hat q v_1 - \bar v_1 \hat q v_2|^2 },
\end{equation}
where $q$ is a pure imaginary quaternion, $iq$ is the VEV of $\phi$ and $\hat q$ is the unit quaternion in the direction of $q$. As with the metric, the first two terms are the potentials for each instanton lump individually. The final term describes their interaction.

It is well known that the moduli space has singularities corresponding to instantons of zero size. These can be understood as the conical singularities where the moduli space is quotiented by discrete symmetries in the ADHM data. By fixing the parameter $\alpha = 0$ we have broken the continuous $\O(2)$ symmetry of the ADHM data in equation \eqref{eq:rotation of ADHM data}. However, there still remains a discrete subgroup of symmetries. The moduli space is quotiented by these symmetries since they identify equivalent parameterisations of the ADHM data which correspond to the same gauge field. The moduli space therefore has an orbifold structure with conical singularities at the fixed points of these symmetries.

Consider a transformation of the ADHM data where $R$ in equation \eqref{eq:rotation of ADHM data} is a rotation matrix. This gives the equivalent parameterisation:
\begin{align}
\begin{split} \label{eq:discrete rotation of ADHM data}
\tilde v_1 &= v_1 c - v_2 s, \qquad \tilde v_2 = v_1 s + v_2 c, \\
\tilde \tau &= (c^2 - s^2) \tau - 2 cs \,\sigma, \\
\tilde \sigma &= (c^2 - s^2) \sigma + 2cs\, \tau,
\end{split}
\end{align}
where $c = \cos(\theta)$ and $s = \sin(\theta)$. If $R$ is a reflection matrix instead then,
\begin{align}
\begin{split} \label{eq:discrete reflection of ADHM data}
\tilde v_1 &= v_1 c + v_2 s, \qquad \tilde v_2 = v_1 s - v_2 c, \\
\tilde \tau &= (c^2 - s^2) \tau + 2 cs \,\sigma, \\
\tilde \sigma &= - (c^2 - s^2) \sigma + 2cs\, \tau.
\end{split}
\end{align}
For these to leave $\alpha = 0$ invariant we must have either
\begin{equation}
c^2 - s^2 = 0, \quad \text{or} \quad cs = 0,
\end{equation}
so that the remaining discrete symmetries are given by rotations or reflections with angle $\theta = \tfrac{n \pi}{4}$ for $n = 0, \ldots, 7$. These form the Dihedral group of order $16$.

Let us consider the action of these remaining symmetries based on their physical interpretation.
\begin{enumerate}
  \item{
    $\mathbf{c = \pm 1, s = 0}$. Under these symmetries, $v_1$ and $v_2$ are unchanged or negated. Consider a reflection with $c=-1$ where $v_1$ goes to $-v_1$. The moduli space is therefore quotiented by a $Z_2 = \{-1, 1\}$ symmetry with a fixed point at $v_1 = 0$. The moduli space has the topology of a cone around the point $v_1 = 0$, which is a conical singularity. The same arguments apply to the point $v_2 = 0$ when $c = 1$. These singularities correspond to an instanton shrinking to zero size.
  }
  \item{
    $\mathbf{c = 0, s = \pm 1}$. Under these symmetries, $v_1$ and $v_2$ swap roles, with a possible change in sign. The parameter describing the instanton separation, $\tau$, is negated. This corresponds to a relabelling of the instantons so that the instanton described by $v_1$ is now described by $v_2$ and vice-versa. The fixed points of these symmetries are when $v_1$ and $v_2$ are equal up to a sign, and $\tau = 0$. The singularities at these fixed points are the same singularities described above, but in a different parameterisation of the moduli space. To see this, consider the following two equivalent parameterisations,
    \begin{equation}
    \tau = \varepsilon, \quad \sigma = i, \quad v_1 = 1 + i \varepsilon \quad \text{and} \quad v_2 = 1 - i \varepsilon,
    \end{equation}
    and
    \begin{equation}
    \tilde \tau = i, \quad \tilde \sigma = \varepsilon, \quad \tilde v_1 = \sqrt{2} \quad \text{and} \quad \tilde v_2 = \sqrt{2} i \varepsilon,
    \end{equation}
    which are identified under a reflection with $\theta = \tfrac{\pi}{4}$. As $\varepsilon \rightarrow 0$, the first of these parameterisations approaches the singularity here. However this is equivalent to the second parameterisation which approaches the zero size instanton singularity mentioned above.
  }
  \item{
    $\mathbf{c = \pm \tfrac{1}{\sqrt{2}}, s = \pm \tfrac{1}{\sqrt{2}}}$. These combine $v_1$ and $v_2$ in a linear combination, and swap the roles of $\tau$ and $\sigma$. The physical interpretation of this symmetry is less obvious but we will discuss it further in Section \ref{sec:right angled scattering} and see that it is responsible for right angled scattering.
  }
\end{enumerate}

From the string theory viewpoint, the zero size singularities arise from transition between the Higgs and Coloumb branches of the D$4$-D$0$ brane system. It is natural that the world volume description should break down at this point.

The moduli space has $12$ parameters excluding the centre of mass and integrating the equations of motion on this full space is numerically expensive.  We can reduce the range of parameters that we need to consider at once by finding geodesic submanifolds of the moduli space\footnote{We use the term geodesic loosely to also include motion on the moduli space in the presence of a potential.}. If our initial conditions lie within a geodesic submanifold then the evolution will remain within the submanifold for all time. A simple way of finding geodesic submanifolds is as the fixed points of symmetries of the metric and potential.

To be able to see the symmetries more explicitly, let us write the metric in the unexpanded form
\begin{equation}
\dif s^2 = 8 \pi^2 \del{ \dif v_1 \cdot \dif v_1 + \dif v_2 \cdot \dif v_2 + \dif \tau \cdot \dif \tau + \dif \sigma \cdot \dif \sigma - \frac{\dif k^2}{N_A} },
\end{equation}
where
\begin{equation}
N_A = |v_1|^2 + |v_2|^2 + 4(|\tau|^2 + |\sigma|^2),
\end{equation}
and
\begin{equation}
\dif k = \bar v_1 \dif v_2 - \bar v_2 \dif v_1 + 2 (\bar \tau \dif \sigma - \bar \sigma \dif \tau),
\end{equation}
is real, recalling that $\sigma$ is given by
\begin{equation}
\sigma = \frac{\tau}{4|\tau|^2} (\bar v_2 v_1 - \bar v_1 v_2).
\end{equation}
The potential is
\begin{equation}
V = 8 \pi^2 |q|^2 \del{ |v_1|^2 + |v_2|^2 - \frac{1}{N_A}|\bar v_2 \hat q v_1 - \bar v_1 \hat q v_2|^2 }.
\end{equation}

The first symmetry that we will consider is conjugation by a quaternion,
\begin{equation}
  v_1 \rightarrow p v_1 \bar p, \quad 
  v_2 \rightarrow p v_2 \bar p, \quad 
  \tau \rightarrow p \tau \bar p, \quad \text{under which} \quad \sigma  \rightarrow p \sigma \bar p,
\end{equation}
where $p$ is a unit quaternion. This is a symmetry of the metric for any $p$ but is only a symmetry of the potential when $p \in \spn \{1, q\}$. The fixed points of this $\U(1)$ action are therefore a half-dimensional moduli space. This is in fact the Hanany-Tong correspondence \cite{Hanany:2003hp} (and see also \cite{Tong:2005un}) between instanton and vortex moduli spaces. Clearly we are still left with the size parameters so it appears that the correspondence is with the semi-local vortices in the $\U(1)$ theory with 2 flavours. More precisely, since we do not have a non-commutative deformation for the instantons, we are in the strong gauge coupling limit for the 2+1 dimensional Yang-Mills-Higgs theory with vortices. This means that this half-dimensional moduli space describes the moduli space of charge $2$ $\C P^1$ (or $\O(3)$) $\sigma$-model lumps. A non-zero scalar VEV in the 4+1 dimensional theory results in a potential on the moduli space, and this potential also appears in the vortex (and $\sigma$-model lump) moduli space where the scalar VEV corresponds turning on masses for the hypermultiplets.

Without loss of generality we can take $q$ to be in the direction $e_1$ so that the geodesic submanifold consists of the points when $v_1$, $v_2$ and $\tau$ are only complex valued, with their $e_2$ and $e_3$ components set to zero. This describes the instantons moving in a two dimensional plane of the full four dimensional space, with each instanton having a gauge orientation in the remaining unbroken $\U(1)$ given by the complex phase of $v_1$ and $v_2$. The metric simplifies on this subspace since many of the terms vanish when restricted to only complex values. It is convenient to parameterise this complex submanifold by polar coordinates \cite{Peeters:2001np},
\begin{align}
v_1 &= \rho_1 (e_4 \cos \theta_1 + e_1 \sin \theta_1), \\
v_2 &= \rho_2 (e_4 \cos \theta_2 + e_1 \sin \theta_2), \\
\tau &= \omega (e_4 \cos \chi + e_1 \sin \chi).
\end{align}
The angles can be combined into a relative and overall gauge rotation,
\begin{align}
\phi   &= \theta_1 - \theta_2, \\
\Theta &= \theta_1 + \theta_2.
\end{align}
The metric and potential on this complex submanifold are then
\begin{align}
\begin{split} \label{eq:polar coordinate metric}
\frac{\dif s^2}{8\pi^2} &= \dif \rho_1^2 + \dif \rho_2^2 + \tfrac{1}{4} ( \rho_1^2 + \rho_2^2 ) (\dif \Theta^2 + \dif \phi^2) + \tfrac{1}{2} ( \rho_1^2 - \rho_2^2 ) \, \dif \Theta \, \dif \phi \\
&\qquad + \,  \frac{1}{4\omega^2} \left( \dif\,(R \sin \phi) \right)^2 + \left( 1 + \frac{1}{4\omega^4} R^2 \sin^2 \phi \right) (\dif \omega^2 + \omega^2 \dif \chi^2) \\
&\qquad -\frac{1}{2 \omega^4} \left( R \sin^2\phi\, ( \rho_1 \dif \rho_2 + \rho_2 \dif \rho_1 ) + R^2 \cos{\phi} \sin{\phi} \, \dif \phi \right) \omega \, \dif \omega \\
&\qquad -\frac{1}{N_A} \left( \cos{\phi}\, ( \rho_1 \dif \rho_2 - \rho_2 \dif \rho_1 ) +  R \sin \phi \, (\dif \Theta - 2 \dif \chi) \right)^2,
\end{split}
\end{align}
and
\begin{equation}
\frac{V}{8 \pi^2} =  |q|^2 \left( \rho_1^2 + \rho_2^2 - \frac{4}{N_A} R^2 \sin^2 \phi \right),
\end{equation}
where
\begin{equation}
N_A = \rho_1^2 + \rho_2^2 + 4 \omega^2 + \frac{R^2}{\omega^2} \sin^2 \phi, \qquad R \equiv \rho_1 \rho_2.
\end{equation}
Note that the metric has no functional dependence on $\Theta$ and $\chi$. These correspond to the overall gauge rotation and spatial rotation of the instantons respectively. This is to be expected as these are symmetries of the full field theory and so descend to Killing vectors on the moduli space. In this parameterisation it is clear that $V$ is the square of the Killing vector corresponding to rotations by $\Theta$, as described in Section \ref{sec:moduli space potential}.

We can further restrict to a submanifold of this complex submanifold by relating the two instantons' sizes and gauge angles. Consider the symmetry,
\begin{equation}
v_1 \rightarrow e_1 v_2, \quad v_2 \rightarrow - e_1 v_1.
\end{equation}
The fixed points of this are when $v_1 = e_1 v_2$, or in our polar coordinate parameterisation,
\begin{equation}
\rho_1 = \rho_2, \quad
\theta_1 = \theta_2 - \tfrac{\pi}{2}.
\end{equation}
On this submanifold we will drop the subscripts on $\rho$ and $\theta$ since they are unnecessary. The metric and potential are
\begin{align}\label{eq:orthogonal metric}
\begin{split}
\frac{\dif s^2}{8 \pi^2} &= 2 \dif \rho^2 + 2 \rho^2 \dif \theta^2 + \frac{\rho^2}{\omega^2} \dif \rho^2 + \del{1 + \frac{\rho^4}{\omega^4}}\del{\dif \omega^2 + \omega^2 \dif \chi^2} \\
&\qquad -\frac{\rho^3}{\omega^3} \dif \rho \dif \omega - \frac{4}{N_A} \rho^4 \del{\dif \theta - \dif \chi}^2,
\end{split} \\
\frac{V}{8 \pi^2} &= q^2 \del{ 2 \rho^2 - \frac{4}{N_A} \rho^4 }, \label{eq:orthogonal potential} \\
N_A &= 2\rho^2 + 4\omega^2 + \frac{\rho^4}{\omega^2}.
\end{align}

We will see that this geodesic submanifold has very similar properties to the charge two Q-lump moduli space for a deformed $\O(3)$ $\sigma$-model, studied by Leese \cite{Leese1991283}. The deformation leads to a potential on the moduli space of $\sigma$-model lumps. The reduction from a 6- to 4-dimensional moduli space is because in the 2+1 dimensional theory, some of the moduli correspond to non-normalisable modes and so are frozen. The Hanany-Tong correspondence does not remove these modes which include the relative size and orientation of the 2 lumps \cite{Ward:1985ij, Leese:1992fn}. However, this further reduced moduli space removes exactly those modes, hence the qualitative similarities.

We can also consider a related symmetry where
\begin{equation}
v_1 \rightarrow v_2, \quad v_2 \rightarrow v_1.
\end{equation}
The fixed points of this are when $v_1 = v_2$, but the metric and potential on this submanifold just reduce to that of two non-interacting charge one dyonic instantons. We can conclude that when the dyonic instanton's gauge alignments are parallel they do not interact with each other.

There is another interesting symmetry which will be relevant to our discussion of localised charge two dyonic instantons. We swap $\tau$ and $\sigma$ with a quaternionic phase,
\begin{equation}
  \tau \rightarrow p \sigma, \quad \sigma \rightarrow -p \tau,
  \end{equation}
  where $p$ is a purely imaginary unit quaternion. This has a fixed point when $\tau = p \sigma$ so that $|\tau| = |\sigma|$. We will see in Section \ref{sec:right angled scattering} that this corresponds to the instantons being coincident. The magnitude of $\tau$ must be fixed by
  \begin{equation}
  |\tau|^2 = \tfrac{1}{4}|\bar v_2 v_1 - \bar v_1 v_2|,
\end{equation}
which will remove a parameter from the metric on this submanifold. This submanifold is also invariant under the symmetries of ADHM data and is a natural boundary on the fundamental domain of the moduli space.

\section{Dynamics of a single instanton} \label{sec:dynamics of a single instanton}

Before looking at interacting instantons it will be useful to briefly review the dynamics of a single instanton.

The effective action for a single dyonic instanton rotating in only one direction in the gauge group is,
\begin{equation}
S = 8 \pi^2 \int \dif t \, \dot \rho^2 + \rho^2 \dot \theta^2 - |q|^2 \rho^2,
\end{equation}
where $\rho$ is the size of the instanton and $\theta$ is its $\U(1)$ gauge angle. This can be worked out directly from the inner product of zero-modes of the 't Hooft ansatz \cite{Peeters:2001np} or from the ADHM data as in Appendix \ref{ap:calculation of metric}. The equation of motion for the gauge angle is a conservation law for gauge angular momentum,
\begin{equation}
\rho^2 \dot \theta = L,
\end{equation}
where $L$ is some constant. The equation of motion for $\rho$ is
\begin{equation}
\ddot \rho - \rho \dot \theta^2 + |q|^2 \rho = 0.
\end{equation}
We can replace $\dot \theta$ by the angular momentum so that
\begin{equation}
\ddot \rho - \frac{L^2}{\rho^3} + |q|^2 \rho = 0.
\end{equation}

In the absence of a potential ($|q| = 0 $), pure instantons suffer from a slow-roll instability where a small perturbation to the static instanton will result in the instanton spreading out at a constant velocity. Eventually the instanton will be spread over the entire space or hit the zero size singularity. We can easily see this behaviour on the moduli space since the metric in the effective action is flat and the equation of motion for $\rho$ becomes $\ddot \rho = 0$ in the absence of any angular velocity. 

The effective action for a dyonic instanton includes a potential term which stabilises the lumps at a fixed size. We can see from the equation of motions that when $\dot \theta = v$ the instanton size and rotational velocity remain constant. This describes a static dyonic instanton which satisfies the BPS equations and equations of motion exactly. The apparent motion on the moduli space is due to the coordinate transformation that we made in equation \eqref{eq:moduli space coordinate transformation}.

If we think of this motion as a particle rolling around a potential like a marble in a bowl then it is clear that this system is now stable to perturbations in the instanton's size; a small initial velocity for $\rho$ sets up an oscillation around the initial value of $\rho$, but it will not increase indefinitely. The upper and lower bounds of the oscillation are proportional to the initial perturbation.

Generically the dyonic instanton will oscillate in size with an amplitude given by \cite{Peeters:2001np},
\begin{equation}
\rho = \sqrt{ A \sin(2|q|(t + t_0))) + \sqrt{\frac{L^2}{|q|^2} + A^2} }.
\end{equation}
The smaller the initial angular velocity, the less angular momentum the instanton has and the closer is comes to zero size. The larger the initial change in size, the larger the amplitude of the oscillation and again the closer it will come to zero size. The instanton can oscillate out to arbitrary size for a sufficiently large initial $\dot \rho$ but will always turn around before reaching $\rho = 0$ unless the angular momentum is zero.

\section{Instanton scattering} \label{sec:right angled scattering}

Right angled scattering is a common feature in soliton systems and in this section we see that instantons are no exception. Recall that right angled scattering of monopoles can be understood from the conical structure of the monopole moduli space under an identification of the two incoming monopoles. The instanton moduli space is more complicated and does not have this simple conical structure but we can still understand right angled scattering on it by considering the underlying symmetries of the ADHM data which quotient the moduli space.

When the parameter $\tau$ in the ADHM data is large it describes two well separated instantons at $\tau$ and $-\tau$, as can be seen in Figure \ref{fig:effect of tau when well separated}. This interpretation is less clear when $\tau$ is of a similar magnitude to $\sigma$. As a first observation, we note that the instantons are coincident when $|\tau| = |\sigma|$. Coincident instantons are axially symmetric as in Figure \ref{fig:effect of tau when coincident}. When the magnitude of $\tau$ is less than the magnitude of $\sigma$ the instantons separate again but at right angles, as seen in Figures \ref{fig:effect of tau when scattered 1} and \ref{fig:effect of tau when scattered 2}. Clearly $\sigma$ also plays an important role in describing the separation.

\begin{figure}
  \begin{subfigure}[b]{0.3\textwidth}
    \includegraphics[width=\textwidth]{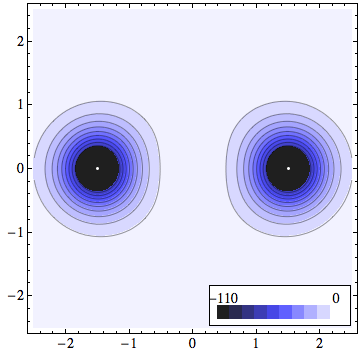}
    \subcaption{$|\tau| = 1.5$ ($|\sigma| = \tfrac{1}{3}$)} \label{fig:effect of tau when well separated}
  \end{subfigure} \hfill
  \begin{subfigure}[b]{0.3\textwidth}
    \includegraphics[width=\textwidth]{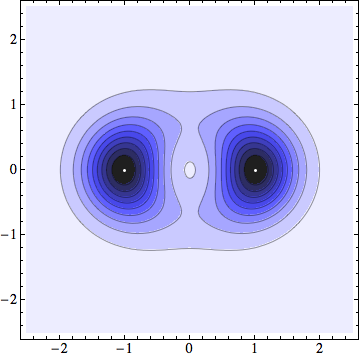}
    \subcaption{$|\tau| = 1$ ($|\sigma| = \tfrac{1}{2}$)}
  \end{subfigure} \hfill
  \begin{subfigure}[b]{0.3\textwidth}
    \includegraphics[width=\textwidth]{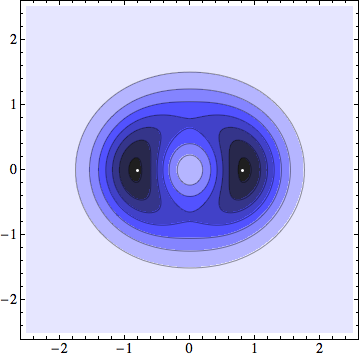}
    \subcaption{$|\tau| = 0.8$ ($|\sigma| = 0.625$)}
  \end{subfigure} \\[10pt]
  \begin{subfigure}[b]{0.3\textwidth}
    \includegraphics[width=\textwidth]{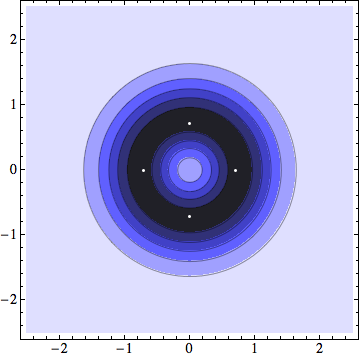}
    \subcaption{$|\tau| = \tfrac{1}{\sqrt{2}}$ ($|\sigma| = \tfrac{1}{\sqrt{2}}$)} \label{fig:effect of tau when coincident}
  \end{subfigure} \hfill
  \begin{subfigure}[b]{0.3\textwidth}
    \includegraphics[width=\textwidth]{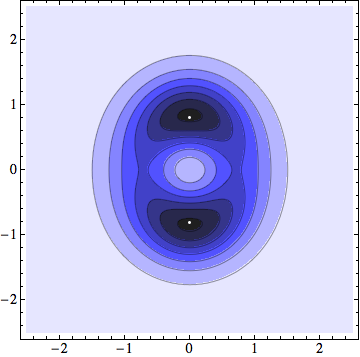}
    \subcaption{$|\tau| = \tfrac{1}{2} $ ($|\sigma| = 1$)} \label{fig:effect of tau when scattered 1}
  \end{subfigure} \hfill
  \begin{subfigure}[b]{0.3\textwidth}
    \includegraphics[width=\textwidth]{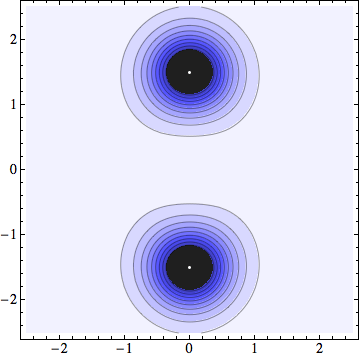}
    \subcaption{$|\tau| = \tfrac{1}{3}$ ($|\sigma| = 1.5$)} \label{fig:effect of tau when scattered 2}
  \end{subfigure}
  \caption{The topological charge density of a charge two instanton at various values of $\tau$. Each figure shows the values of the charge density on the complex plane, at zero in the $e_2$ and $e_3$ quaternion directions. Each contour shows a fixed value of the charge density with the lighter areas corresponding to lower charge. The instantons have size $\rho = 1$ and have an orthogonal gauge alignment ($\phi = \tfrac{\pi}{2}$). The value of $\tau$ is real. The white dots mark the positions of $\tau$ and $-\tau$ when $|\tau| \ge |\sigma|$ and the positions of $\sigma$ and $-\sigma$ when $|\tau| \le |\sigma|$. }
  \label{fig:effect of tau}
\end{figure}

This relationship between $\tau$ and $\sigma$ can be understood from the symmetries of the ADHM data which swap the roles of $\tau$ and $\sigma$. For example, consider a transformation of the ADHM data as in equation \eqref{eq:discrete reflection of ADHM data} by a reflection with angle $\theta = \tfrac{\pi}{4}$,
\begin{equation} \label{eq:equivalent ADHM parameters}
\tilde \Delta = \begin{pmatrix}
\tfrac{1}{\sqrt{2}} ( v_1 + v_2 ) & \tfrac{1}{\sqrt{2}} ( v_1 - v_2 ) \\
\sigma & \tau \\
\tau   & -\sigma
\end{pmatrix}.
\end{equation}
This symmetry leaves the fields unchanged so these parameters must have an equivalent physical interpretation as those in the original ADHM parameterisation. It follows that $\sigma$ must have an equal claim to describe the instantons' separation.

Recall that $\sigma$ is given by
\begin{equation}
\sigma = \frac{\tau}{4 |\tau|^2} \Lambda, \quad \text{where} \quad \Lambda = \bar v_2 v_1 - \bar v_1 v_2,
\end{equation}
and so has a magnitude inversely proportional to the magnitude of $\tau$. When $\tau$ is large, $\sigma$ is small and $\tau$ provides a good description of the instantons' separation. However, as $\tau$ grows smaller and is of a similar size to $\sigma$ this description breaks down and the instantons are close to coincident. As $\tau$ goes to zero, $\sigma$ grows large and instead takes on the role of the separation.

Right angled scattering occurs because $\sigma$ is equal to $\tau$ multiplied by a pure imaginary quaternion, $\Lambda$. Treating $\sigma$ and $\tau$ as $4$-vectors, their inner product is zero,
\begin{equation}
\sigma \cdot \tau = 0
\end{equation}
so they lie orthogonal to each other. When $\sigma$ takes over the role of the separation, the instantons will be separated at right angles to the previous direction $\tau$. This behaviour is what we see in Figure \ref{fig:effect of tau}.

The direction in which the instantons separate will be determined by the direction of $\sigma$. The rotation compared to their incoming direction, $\tau$, is determined by the purely imaginary quantity $\bar v_2 v_1 - \bar v_1 v_2$. There are three distinct cases although a general scattering may be some combination of these:
\begin{enumerate}
\item{
  \underline{$\tau$, $v_1$ and $v_2$ in the same plane}. When the gauge embeddings of the instantons are in the same plane as their separation, the instantons will scatter orthogonally to $\tau$ in this plane. This is the only situation possible in the complex geodesic submanifold in equation \eqref{eq:polar coordinate metric}.
}
\item{
  \underline{$\tau$ and $v_1$ in the same plane with $v_2$ orthogonal}. When $v_2$ is orthogonal to this plane, the instantons will scatter in the direction of $v_2$. Similarly for $v_1$ and $v_2$ reversed.
}
\item{
  \underline{$\tau$, $v_1$ and $v_2$ all mutually orthogonal}. When the gauge embeddings are orthogonal to each other and the instantons' separation, they will scatter in the remaining direction orthogonal to $\tau$, $v_1$ and $v_2$.
}
\end{enumerate}
Recall that when $v_1$ and $v_2$ are parallel the instantons do not interact. In this case $\sigma$ is zero and $\tau$ always describes their separation. The instantons do not scatter and will instead pass through each other.

As an alternative interpretation, we note that the ADHM data naturally splits into two parts: $\Lambda$, describing the instanton sizes and gauge alignments, and $\Omega$, describing the instantons positions. When an $N \times N$ matrix describes the positions of $N$ D-branes, it is the eigenvalues that actually correspond to the physical positions and for complex valued $\tau$ and $\sigma$ the eigenvalues of $\Omega$ are
\begin{equation}
\pm \sqrt{\tau^2 + \sigma^2}.
\end{equation}
The eigenvalues are approximately equal to $\pm \tau$ when $\tau$ is large and to $\pm \sigma$ when $\tau$ is small. The eigenvalues will be zero when the instantons are coincident and $|\tau| = |\sigma|$. The eigenvalues are rotated by $90^\circ$ in the complex plane when they pass through zero due to a change in sign inside the square root.

We can briefly compare this behaviour to right angled scattering in monopoles. For two monopoles in $\SU(2)$, the moduli space is the Atiyah-Hitchin manifold which has a two dimensional geodesic submanifold corresponding to motion in a plane. This submanifold has the topology of a cone since the system is identical under a rotation by $180^\circ$ around the origin. Head on scattering is described by a geodesic which passes over the vertex of the cone and therefore emerges at $90^\circ$ relative to where it came in. The subspace is smooth at this vertex although the angle jumps by $\tfrac{1}{2}\pi$, as expected from passing through the origin in polar coordinates.

For two instantons, the moduli space also has a geodesic submanifold corresponding to motion in a plane. The metric of this is given in equation \eqref{eq:polar coordinate metric}. This space is still six dimensional and it is not possible to give as simple a description of right-angled scattering as for monopoles. Each instanton has a unique identity and the symmetry under a rotation by $180^\circ$ no longer exists. Instead, we can understand the $90^\circ$ scattering through the  symmetry of the ADHM data as described above and given in equation \eqref{eq:equivalent ADHM parameters}.

So far we have only considered what happens at individual points on the moduli space. From our understanding of the moduli space parameters we expect to see right angled scattering in the geodesic motion of two instantons whenever the magnitude of $\tau$ passes through $|\tau| = |\sigma|$. This is inevitable if $|\tau|$ is decreasing. 

We cannot numerically integrate the equations of motion for a head on collision between two instantons because the symmetry between $\tau$ and $\sigma$ manifests as a discontinuous jump of parameters in the geodesic evolution. This jump is between equivalent parameterisations and so is smooth on the moduli space, but prevents us from finding a numerical solution. However, we can still explore head on collisions by examining the behaviour as the impact parameter goes to zero.

We will set up then instantons with the initial conditions shown in Figure \ref{fig:parameters}. The only parameters which are not shown are the gauge angles, $\theta_1$ and $\theta_2$. The overall gauge angle, $\theta = \theta_1 + \theta_2$ does not have a physical effect on the static instantons, but it is important in the dynamics. The dynamics are invariant under the initial value of $\theta$ so there is no need to specify it when we list initial conditions. Only the relative angle, $\phi = \theta_1 - \theta_2$ needs to be specified. It is convenient to work with a slightly different parameterisation of our initial positions; we introduce the impact parameter, $b$, and the separation along the $x$-axis, $x$, as shown in Figure \ref{fig:parameters}. To consider the scattering of two instantons we start with well separated static instantons and send them towards each other with an initial velocity parallel to the $x$-axis, $\dot x = -v$. Ideally we are interested in the behaviour as the instantons come from infinity but we will settle on $\omega = 50$ as a practical initial separation in our numerical study. Unless otherwise stated we will take the incoming velocity $v$ to be $v = 0.03$ and the initial instanton sizes to be $\rho_1 = \rho_2 = 1$.

\begin{figure}
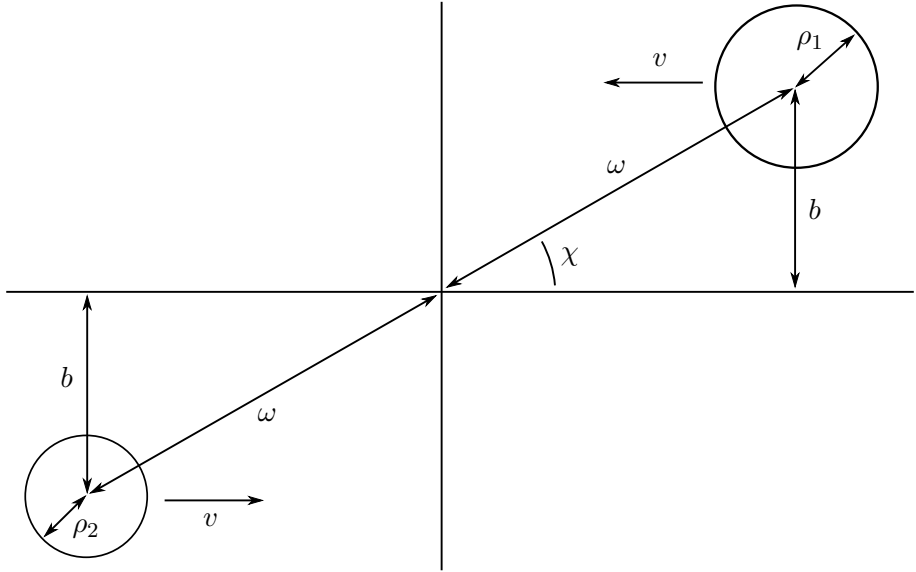

\centering
\include{figures/parameters}
\caption{The physical interpretation of the parameters in the initial conditions of our scattering processes.}
\label{fig:parameters}
\end{figure}

For simplicity we will begin by exploring the geodesic submanifold where the instantons have an identical size and their relative gauge is fixed and orthogonal, $\phi = \tfrac{\pi}{2}$. Figure \ref{fig:instantons scattering with impact parameter} shows the scattering of two instantons on this submanifold with an impact parameter of $b = 0.5$.  The figure shows many snapshots of the instanton's position and sizes at discrete intervals in the evolution. The solid line traces the instantons' positions at $\pm \tau$. The centre of each circle corresponds to the instanton's position and the radius shows the values of, $\rho_1$ and $\rho_2$. These can be interpreted as the instanton sizes although care needs to be taken with this interpretation when the lumps are close together.

After the interaction in Figure \ref{fig:instantons scattering with impact parameter}, the instanton sizes are perturbed and they begin shrink. To the limits of our numerical accuracy the instantons appear to pass through the zero size singularity and emerge with an increasing size, spreading out indefinitely. It may seem concerning that the instantons pass through the singularity on the moduli space, but this is not a generic behaviour. If we move away from the submanifold and give the instantons an initial difference in size or relative gauge angle then they will no longer hit the singularity. The value of $v_1$ and $v_2$ will no longer pass through the origin and the instantons' minimum sizes will be greater than zero.

\begin{figure}
\centering
\includegraphics[width=0.7\textwidth]{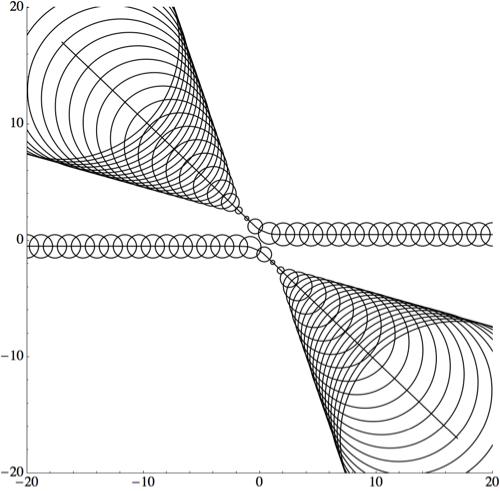}
\caption{The collision of two instantons with an impact parameter of $b = 0.5$. The relative gauge angle begins and remains fixed at $\phi = \tfrac{\pi}{2}$.}
\label{fig:instantons scattering with impact parameter}
\end{figure}

As the impact parameter decreases towards zero the scattering angle increases towards $90^\circ$. Figure \ref{fig:instantons scattering in head on collision} shows the scattering of two instantons with a small impact parameter of $b = 0.01$ where the scattering is at almost exactly $90^\circ$. As mentioned previously, we cannot numerically integrate a direct head on collision due to the discontinuous jump between equivalent parameterisations when the instantons become coincident. This is shown more clearly in Figure \ref{fig:evolution of separation parameters} where the evolution of $|\tau|$ and $|\sigma|$ is shown for impact parameters of $b = 0.1$ and $b = 0.01$. The interpolation between the two becomes increasingly quick as the impact parameter is reduced. The angle $\chi$ also jumps by $\tfrac{\pi}{2}$. This jump can also be seen near the origin in Figure \ref{fig:instantons scattering in head on collision}.

\begin{figure}
\centering
\includegraphics[width=0.7\textwidth]{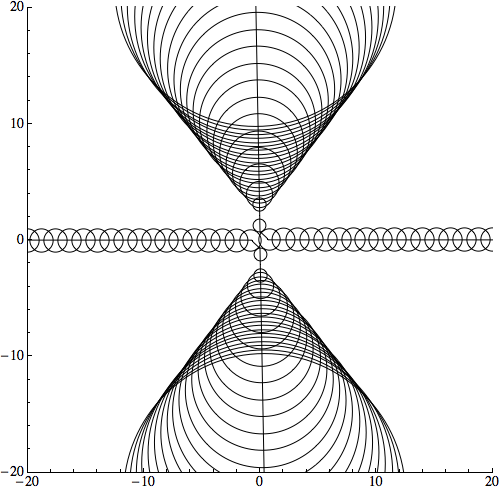}
\caption{The collision of two instantons with an impact parameter of $b = 0.01$. The relative gauge angle begins and remains fixed at $\phi = \tfrac{\pi}{2}$.}
\label{fig:instantons scattering in head on collision}
\end{figure}

\begin{figure}
  \centering
  \begin{subfigure}[b]{0.7\textwidth}
    \includegraphics[width=\textwidth]{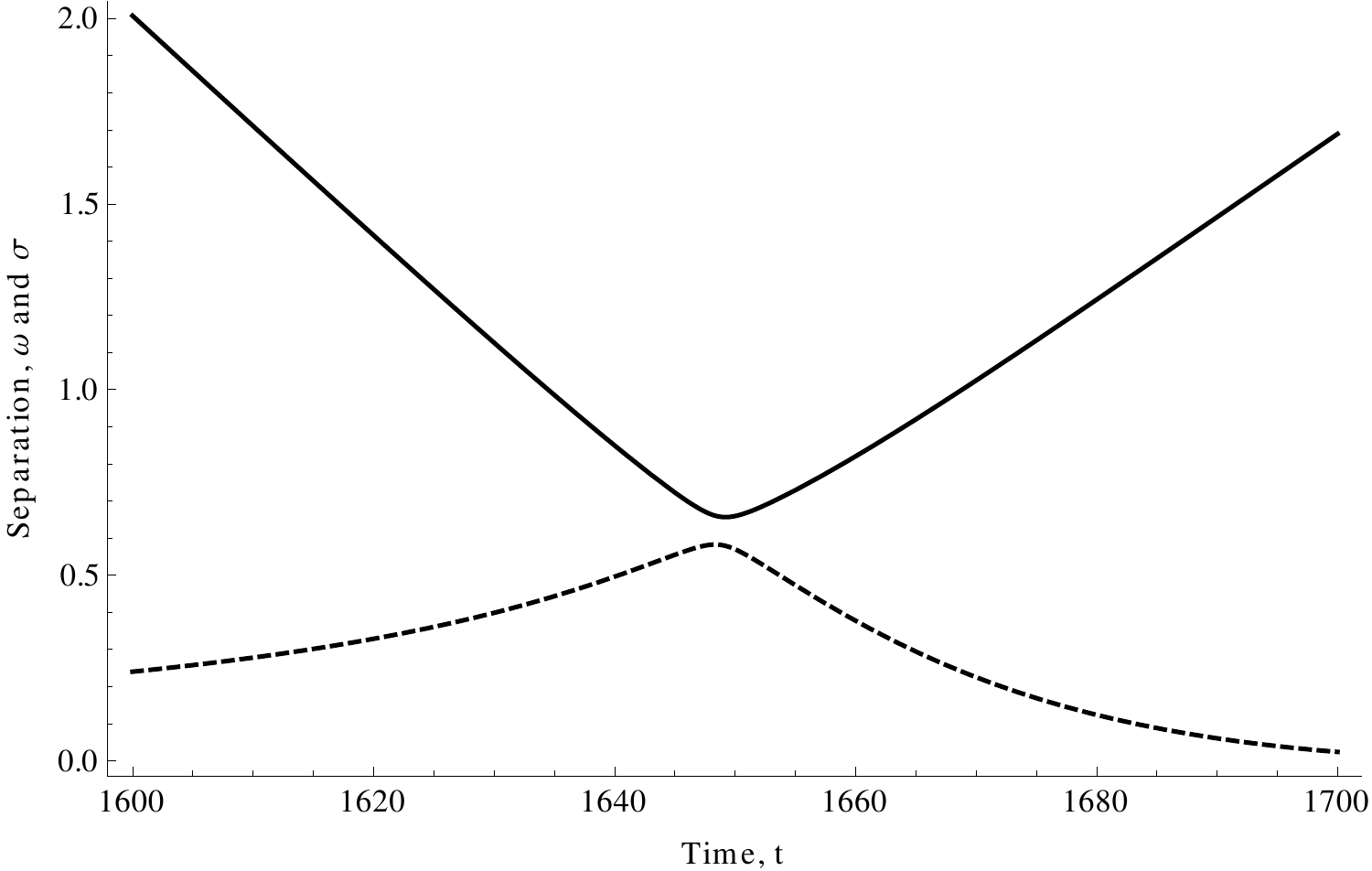}
    \subcaption{Impact parameter $b = 0.1$}
  \end{subfigure} \\[10pt]
  \begin{subfigure}[b]{0.7\textwidth}
    \includegraphics[width=\textwidth]{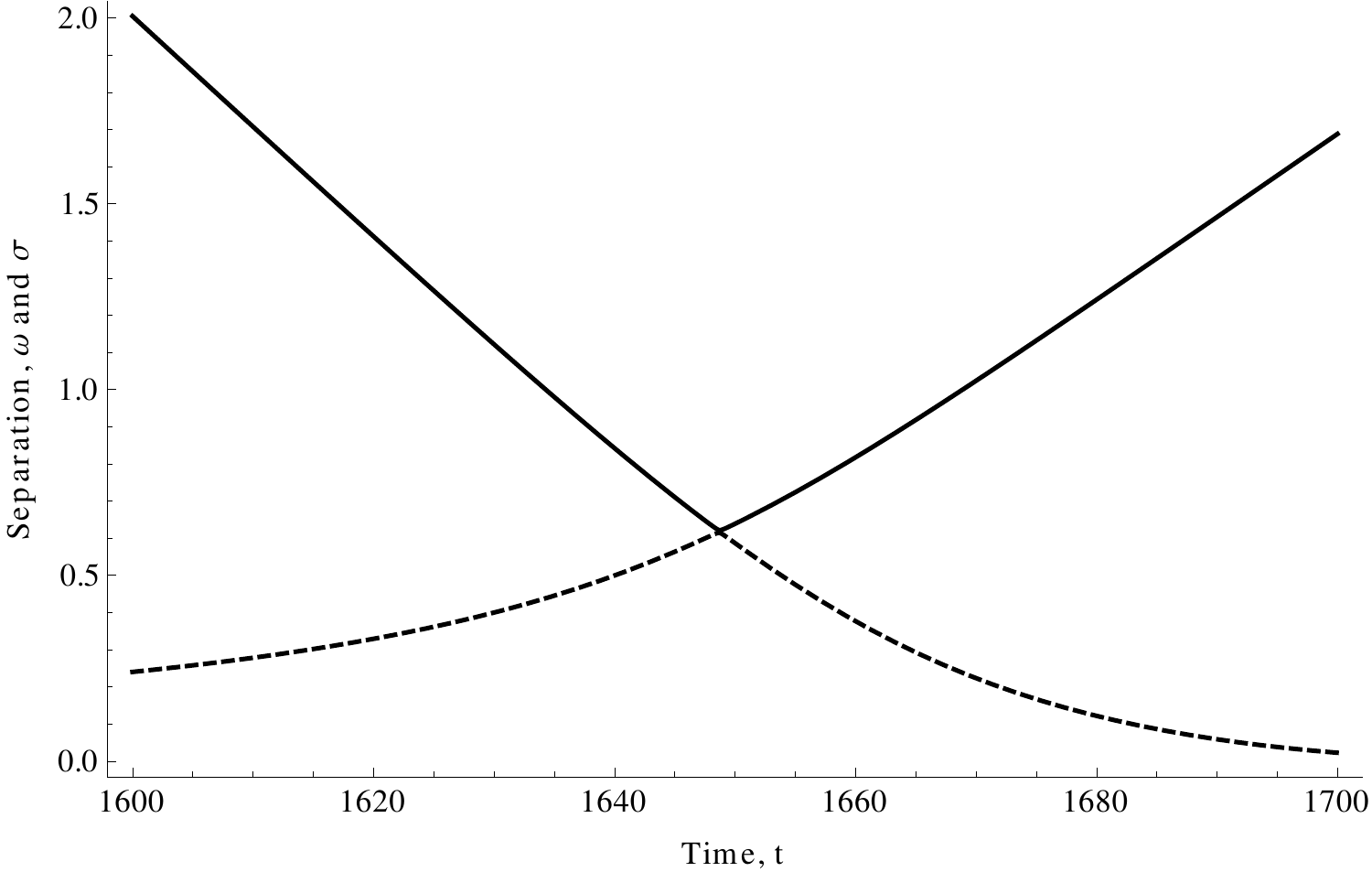}
    \subcaption{Impact parameter $b = 0.01$}
  \end{subfigure} 
  \caption{The evolution of $|\tau|$ (solid) and $|\sigma|$ (dashed) during two collisions of instantons with different impact parameters.}
  \label{fig:evolution of separation parameters}
\end{figure}

Figure \ref{fig:scattering angle vs impact parameter for instantons} shows how the scattering angle varies with impact parameter. The scattering angle clearly tends towards $90^\circ$ as the impact parameter goes to zero. The scattering angle decreases to zero asymptotically as the impact parameter increases.

\begin{figure}
\centering
\includegraphics[width=0.8\textwidth]{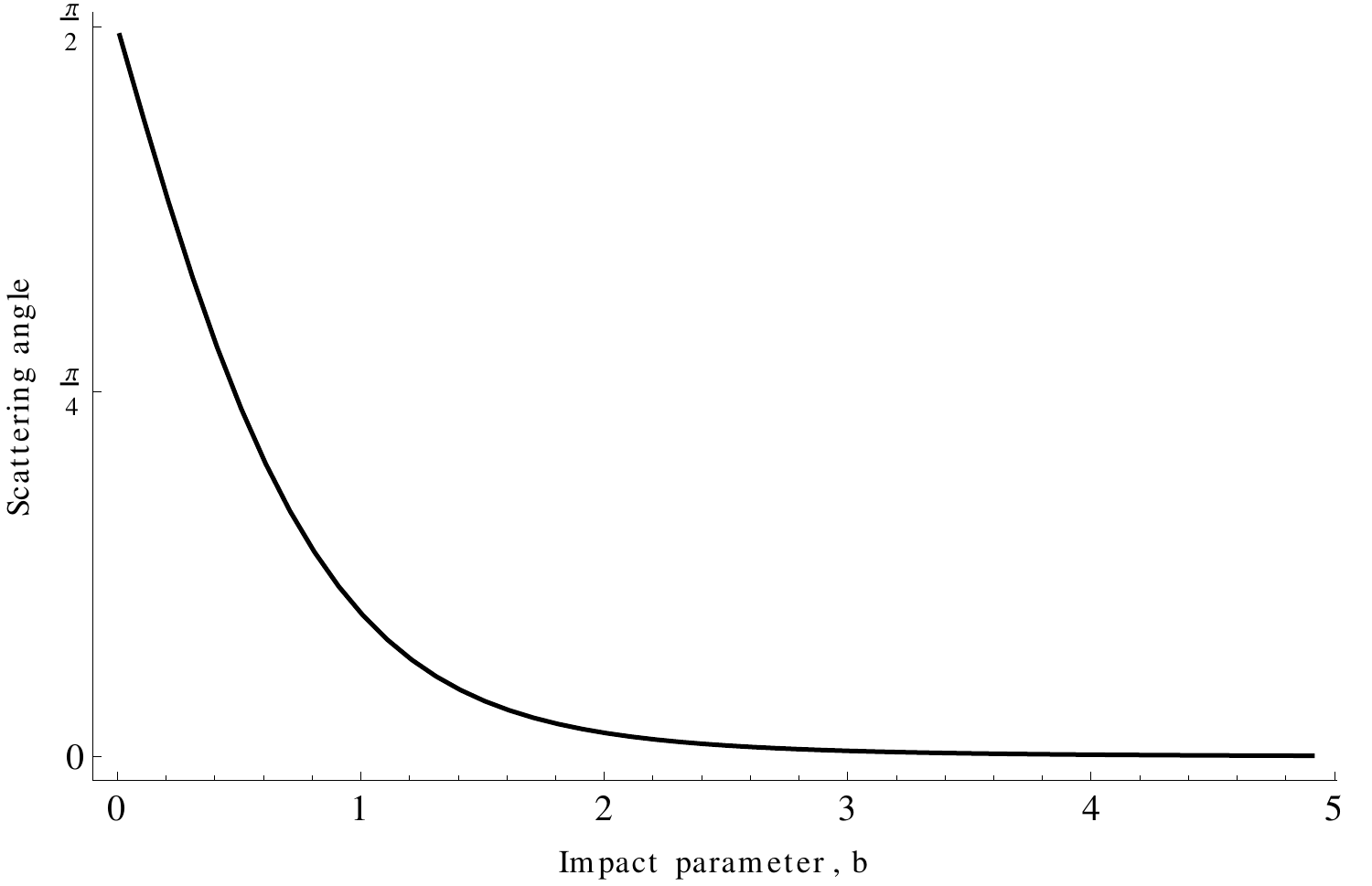}
\caption{The variation of the scattering angle of a collision of two instantons with different impact parameters, $b$.}
\label{fig:scattering angle vs impact parameter for instantons}
\end{figure}

So far we have only considered a subset of the possible initial conditions where the gauge alignment of the two dyonic instantons is orthogonal. We can lift that restriction to explore the effect of this angle however we will still remain in the complex submanifold where the instantons lie and move in a single plane.

We can see from the symmetry in equation \eqref{eq:equivalent ADHM parameters} that the instantons will emerge with sizes,
\begin{equation}
\frac{1}{\sqrt{2}} | v_1 \pm v_2 |.
\end{equation}
The outgoing sizes are only equal when the incoming $v_1$ and $v_2$ are orthogonal, or $\phi = \tfrac{\pi}{2}$. In general they will emerge with different sizes. This description is accurate for immediately before and after the right-angled scattering, however the relation between the asymptotic sizes of the incoming and outgoing instantons is not as clear due to the additional dynamical effects on the size. Figure \ref{fig:instanton scattering with different outgoing sizes} shows the result of a collision with an initial relative gauge angle of $\phi = \tfrac{\pi}{4}$. We see that the instantons now emerge with a different behaviour in their sizes.  The scattering angle is also shallower than when the gauge alignment is orthogonal.

\begin{figure}
\centering
\includegraphics[width=0.7\textwidth]{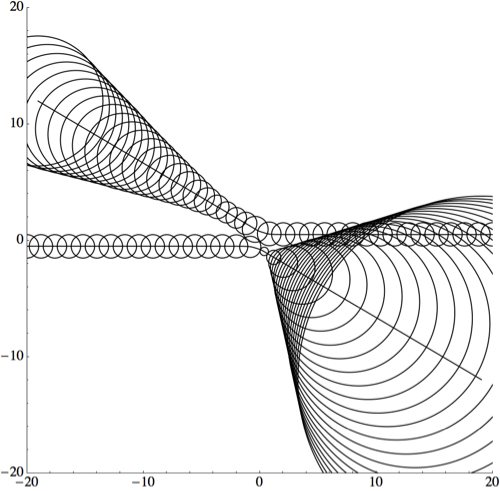}
\caption{The collision of two dyonic instantons with an impact parameter of $b = 0.5$ and an initial relative gauge angle of $\phi = \tfrac{\pi}{4}$.}
\label{fig:instanton scattering with different outgoing sizes}
\end{figure}

Recall that there is another geodesic submanifold corresponding to instantons with parallel gauge angles, $\phi = 0$. In this case the instantons do not interact at all and the metric is flat. The scattering angle is therefore trivially zero. The relative gauge angle between the two instantons therefore gives some measure of the strength of the interaction between the instantons. Figure \ref{fig:scattering angle vs gauge angle for instantons} shows how the scattering angle depends on the initial difference in gauge angle between the two instanton lumps. The strength of the interaction between the instanton lumps depends on the difference in their gauge alignment with the strongest interaction occurring when they are orthogonal. At the other extreme when the gauge alignment is parallel the instantons are completely non-interacting.

\begin{figure}
\centering
\includegraphics[width=0.8\textwidth]{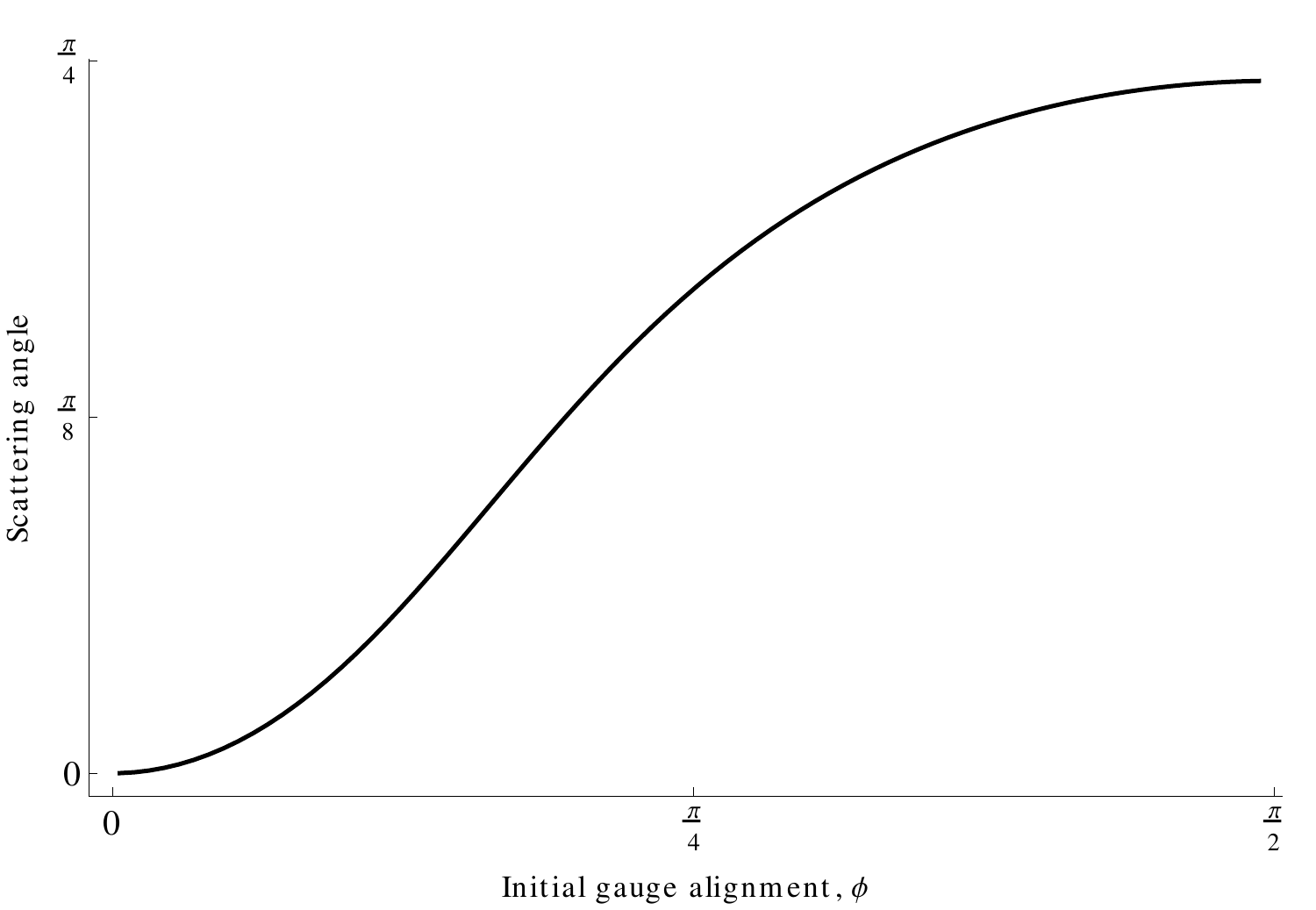}
\caption{The variation of the scattering angle of a collision of two instantons with different initial gauge alignments, $\phi$, and impact parameter $b = -0.5$.}
\label{fig:scattering angle vs gauge angle for instantons}
\end{figure}

Right-angled scattering is a generic feature of two instantons which collide head on, yet we have seen that when the relative gauge angle is zero the instantons do not interact at all. To reconcile the limit of zero gauge angle with right-angled scattering, we note that for small $\phi$ one of the instantons emerges with a much larger size than the other. When $\phi$ is sufficiently small, the large instanton grows in size faster than the instantons separate and so causes them to interact again. This can be seen in Figure \ref{fig:limit of small phi 1} where the initial gauge alignment is $\phi = 0.5$. The instantons scatter at right angles but then continue to interact and rotate slightly for an asymptotic scattering angle of less than $90^\circ$. As the initial gauge angle, $\phi$, goes to zero this effect becomes more pronounced and the asymptotic scattering angle goes to zero, despite the instantons initially scattering at $90^\circ$ when they become coincident. Figures \ref{fig:limit of small phi 2} and \ref{fig:limit of small phi 3} show this effect for $\phi = 0.45$ and $\phi = 0.4$. The  limit of this behaviour is that when $\phi = 0$ the right angled scattering is not apparent and the instantons simply pass through each other after becoming coincident as expected by the moduli space metric.

\begin{figure}
  \centering
  \begin{subfigure}[b]{0.32\textwidth}
    \includegraphics[width=\textwidth]{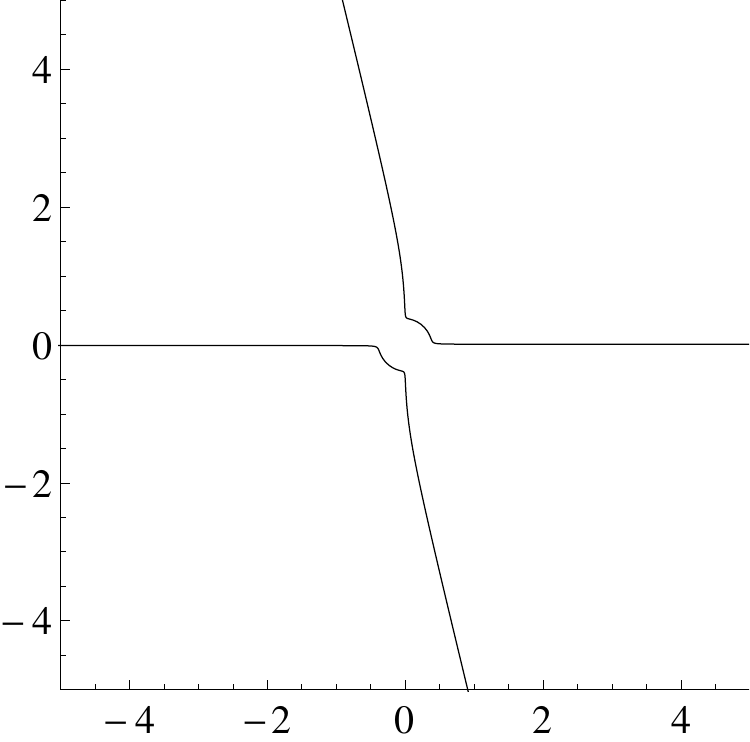}
    \subcaption{Initial gauge angle $\phi = 0.5$}
    \label{fig:limit of small phi 1}
  \end{subfigure} \hfill
  \begin{subfigure}[b]{0.32\textwidth}
    \includegraphics[width=\textwidth]{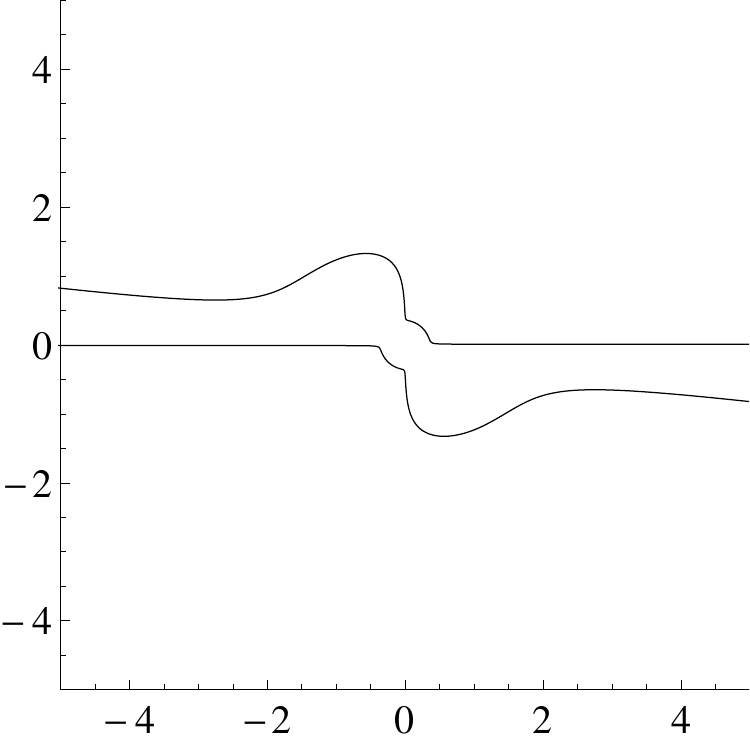}
    \subcaption{Initial gauge angle $\phi = 0.45$}
    \label{fig:limit of small phi 2}
  \end{subfigure} \hfill
  \begin{subfigure}[b]{0.32\textwidth}
    \includegraphics[width=\textwidth]{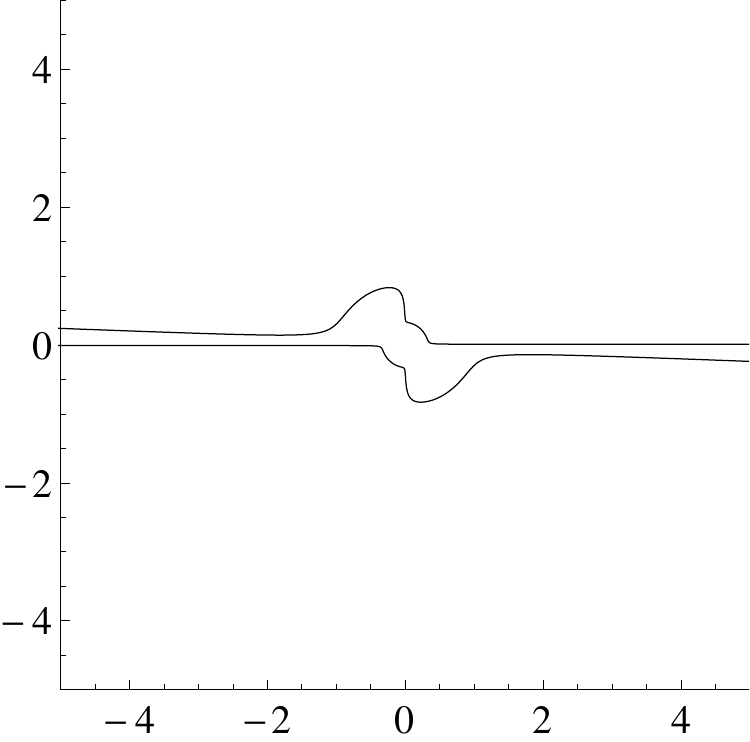}
    \subcaption{Initial gauge angle $\phi = 0.4$}
    \label{fig:limit of small phi 3}
  \end{subfigure}
  \caption{The evolutions of the centres of the instantons in a nearly head on collision ($b = 0.01$) for various values of $\phi$.}
\end{figure}

In the full moduli space, the instantons each have an $\SU(2)$ gauge angle and are free to move in $4$ dimensions. As described previously the scattering direction no longer remains in a plane but depends on the relative angle between the separation of the instantons and the relative gauge alignment of the instantons in $\SU(2)$. Moving beyond the complex subspace into the full moduli space becomes computationally expensive but we have been able to explore a few examples. The complex subspace appears to be stable to small perturbations in the full moduli space so that the discussion above can be safely interpreted in the full moduli space. It would be interesting to explore the scattering behaviour of instantons with their gauge alignment in the full $\SU(2)$ gauge group and not just constrained to the unbroken $\U(1)$. A systematic study of such behaviour is unfortunately beyond our reach at this time.

\section{Dyonic instanton scattering}

The presence of a potential in the effective action for dyonic instantons has a significant effect on their scattering behaviour. In this section we will explore how dyonic instantons behave during head on collisions and with a finite impact parameter. The right angled scattering behaviour of instantons is replaced with a more complex dependence on the potential.

The parameters describing dyonic instantons are identical to those used to describe instantons, as in Figure \ref{fig:parameters}. The only difference is the presence of a potential term in the equations describing their evolution. Figure \ref{fig:head on collision} shows a head on collision between two dyonic instantons. The instantons begin their evolution by moving towards each other along the real axis but they are deflected as they approach either other. The instantons scatter at an unusual angle of just over $122^\circ$ and the radial size of the instantons picks up a small stable oscillation.

\begin{figure}
\centering
\includegraphics[width=0.7\textwidth]{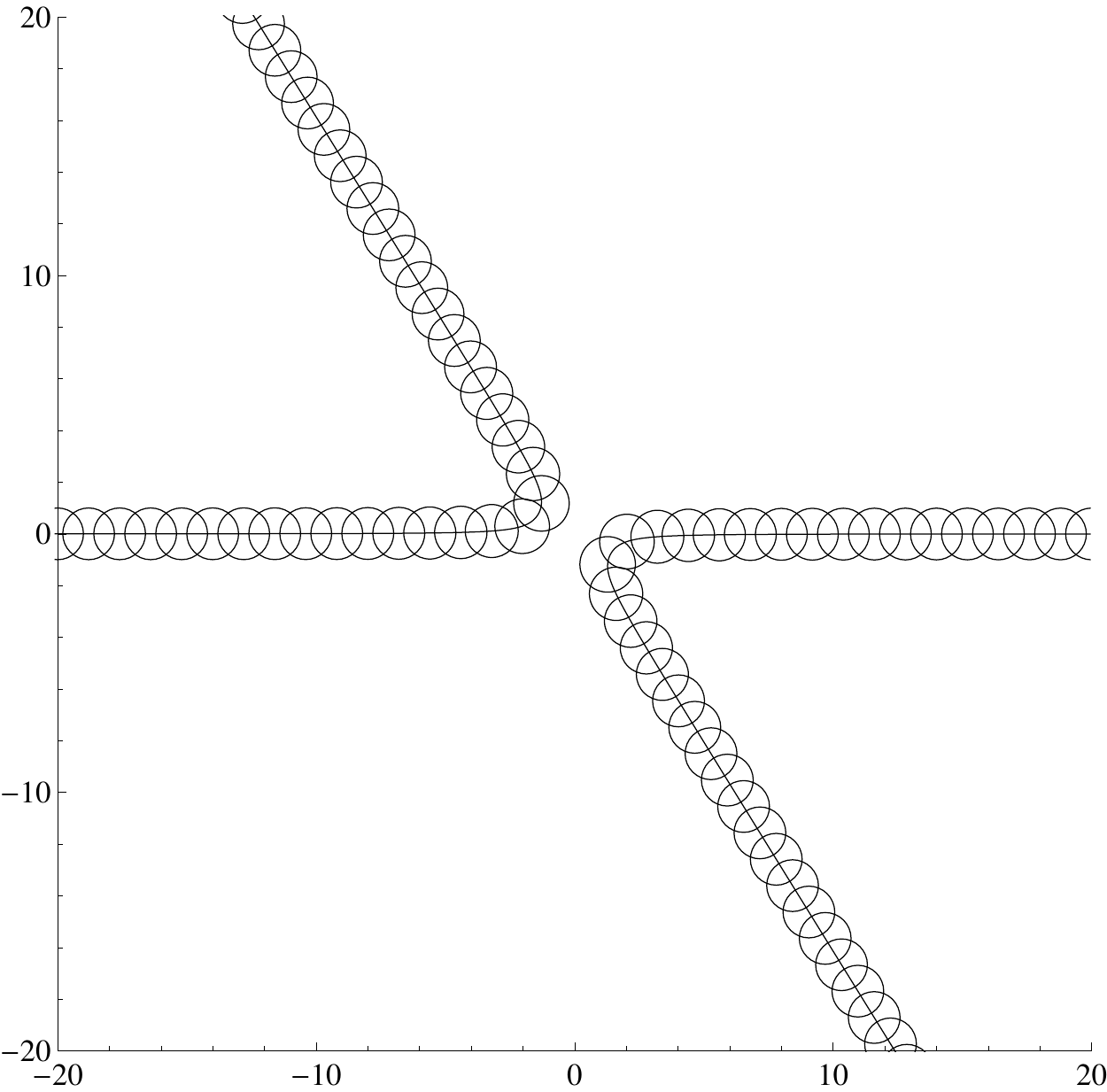}
\caption{A head on collision of two dyonic instantons with the magnitude of the potential at $|q| = 0.1$.}
\label{fig:head on collision}
\end{figure}

Figure \ref{fig:scattering angle vs potential scale} shows the relation between the scattering angle and the magnitude of the potential for dyonic instantons, $q$. As expected, the scattering angle approaches $90^\circ$ as $q$ goes to zero and the system gets closer to describing pure instantons.

\begin{figure}
\centering
\includegraphics[width=0.8\textwidth]{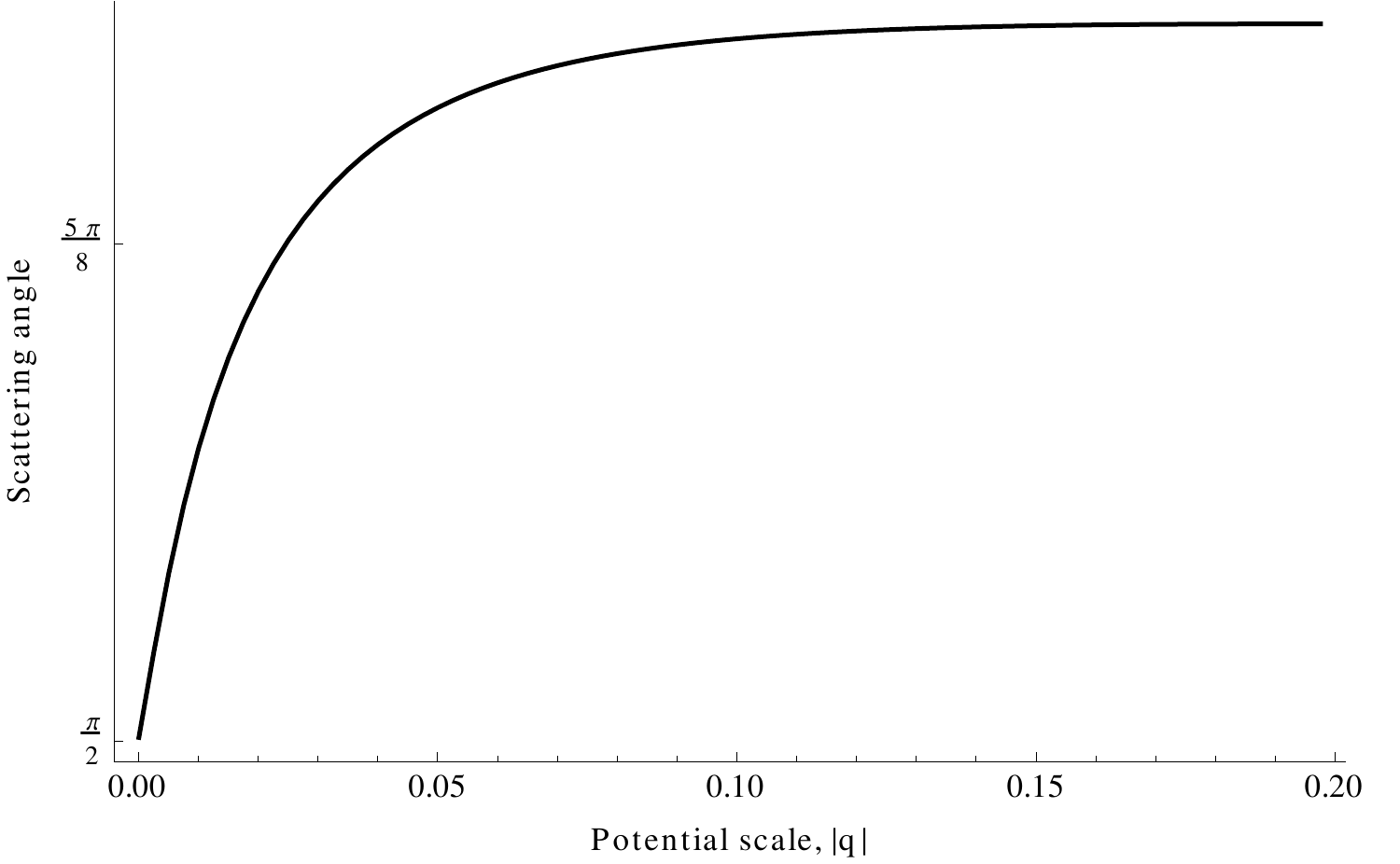}
\caption{The scattering angle of a head on collision of two dyonic instantons with varying values of the potential scale, $|q|$.}
\label{fig:scattering angle vs potential scale}
\end{figure}

When the impact parameter, $b$, is non-zero the dyonic instantons also display a range of interesting behaviour. From the view point of one of the incoming instanton lumps they scatter to their left in a head on collision. If we move their impact parameter in this direction so that $b$ is negative then the instantons continue to repel each other but their scattering angle becomes shallower. Figure \ref{fig:negative impact parameter scattering} shows the scattering of two dyonic instantons with impact parameter $b = -2$ and we see that the scattering angle is much shallower than in the head on collision. Figure \ref{fig:scattering angle vs negative impact parameter for dyonic instantons} shows how the scattering angle depends upon the impact parameter in the negative direction. As the impact parameter is increased in the negative direction the strength of the instanton's interaction decreases and the scattering angle goes to zero. Comparing this to pure instantons in Figure \ref{fig:scattering angle vs impact parameter for instantons} we see that the interaction remains stronger at large impact parameter in the presence of a potential.

\begin{figure}
\centering
\includegraphics[width=0.7\textwidth]{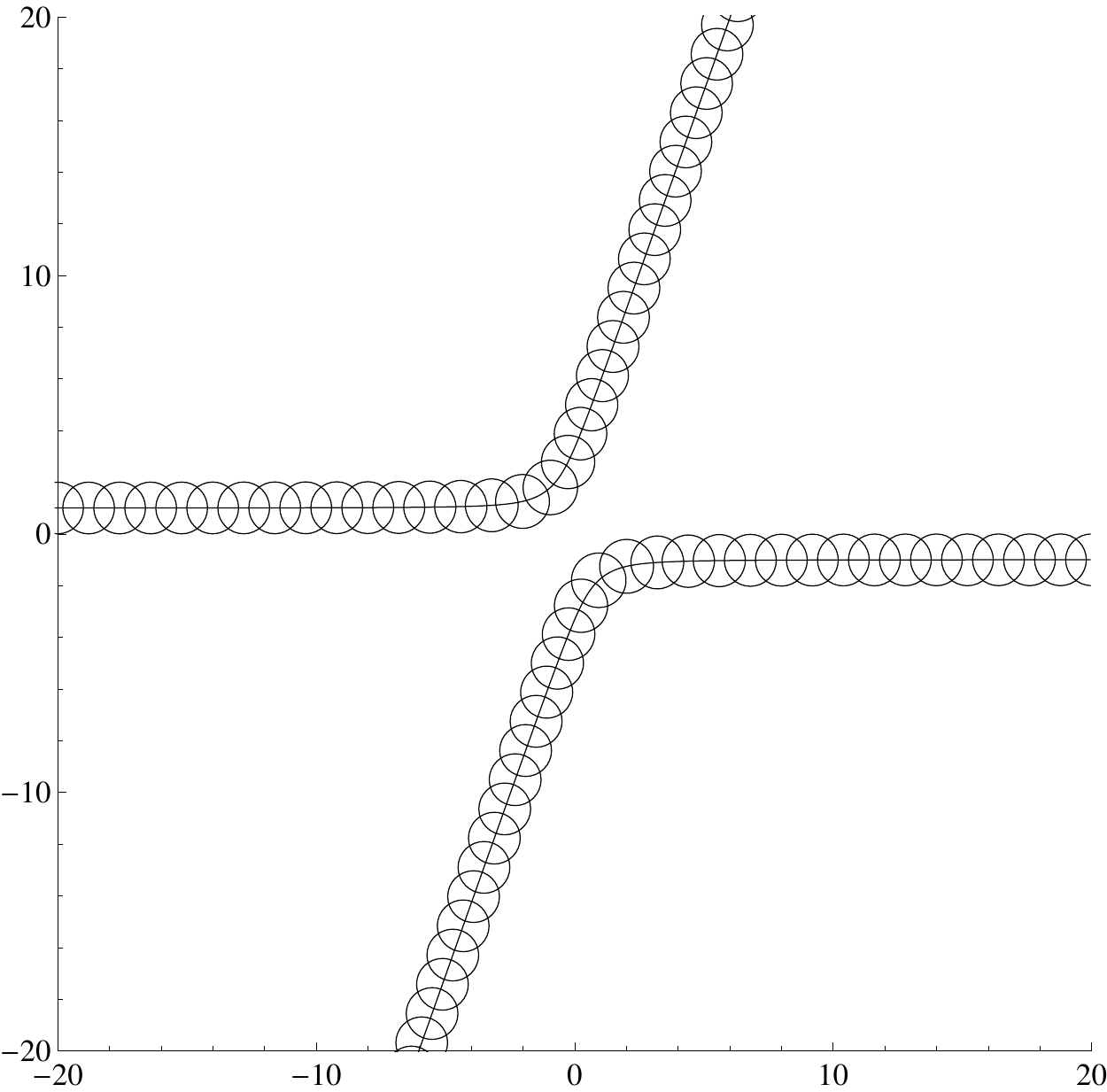}
\caption{The collision of two dyonic instantons with an impact parameter of $b = -1$.}
\label{fig:negative impact parameter scattering}
\end{figure}

\begin{figure}
\centering
\includegraphics[width=0.8\textwidth]{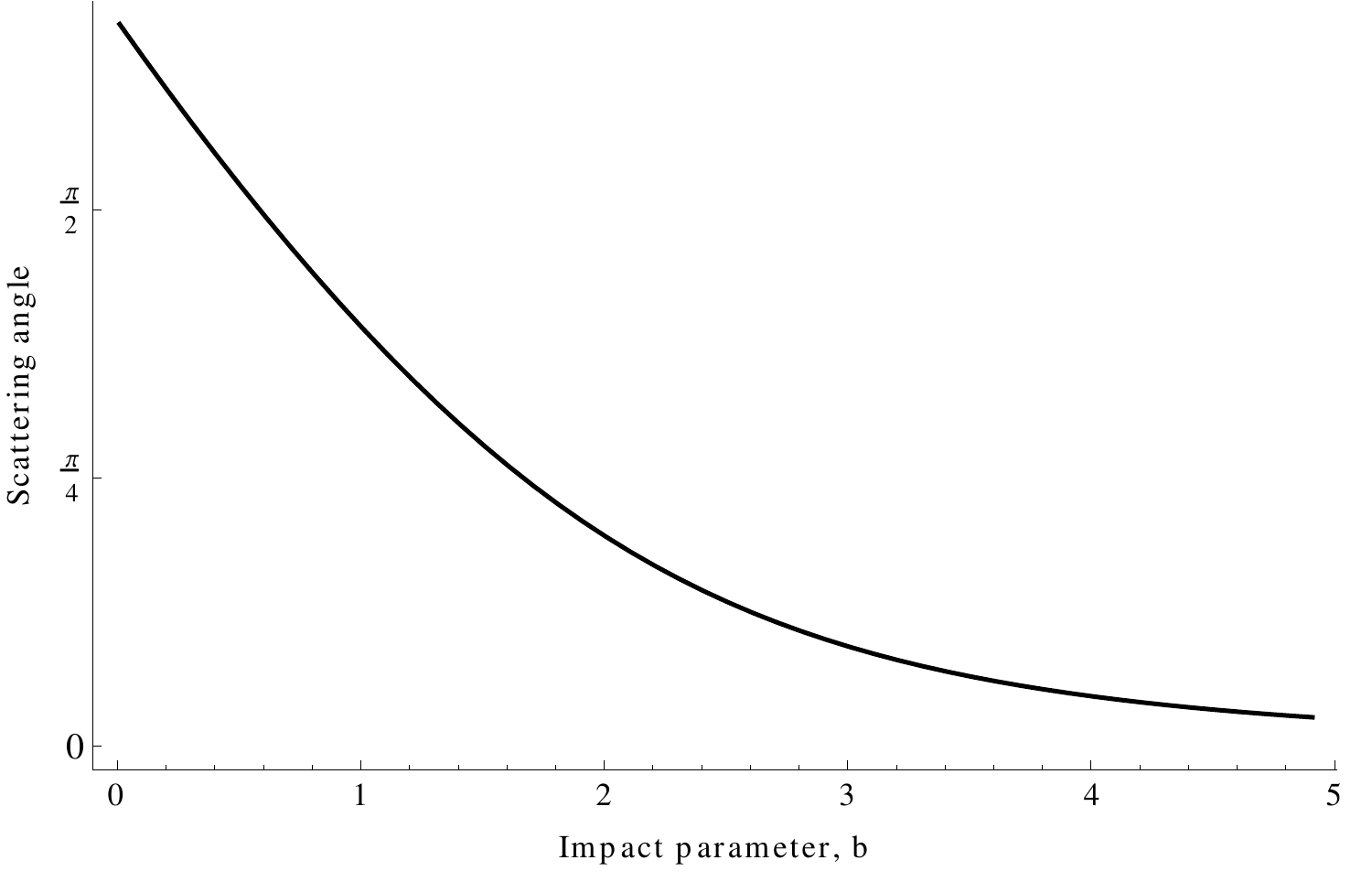}
\caption{The variation of the scattering angle with impact parameter for a collision between two dyonic instantons. The impact parameter, $b$, is in the negative direction.}
\label{fig:scattering angle vs negative impact parameter for dyonic instantons}
\end{figure}

When the scattering angle is positive the behaviour is more interesting. The dyonic instantons are now attracted to each other and it is possible for the instantons to loop around each other before scattering. Figure \ref{fig:positive impact parameter scattering} shows one example of this in detail with an impact parameter of $b = 2.9$. Figure \ref{fig:scattering angle vs positive impact parameter for dyonic instantons} shows how the outgoing angle varies with the impact parameter across a range of different scales of the potential. The jumps in the plots correspond to the instantons losing their identity in the scattering process. This happens whenever the instantons come close to the origin at the same time and form a single symmetrical lump. It becomes meaningless to talk about which outgoing instanton corresponds to which incoming instanton and the jumps by $180^\circ$ are from swapping which parameters are used to label each instanton rather than a physical discontinuity. The tall spikes correspond to scatterings in which the instantons orbit for more than one revolution and so can have an outgoing angle of greater than $360^\circ$.

\begin{figure}
\centering
\includegraphics[width=0.7\textwidth]{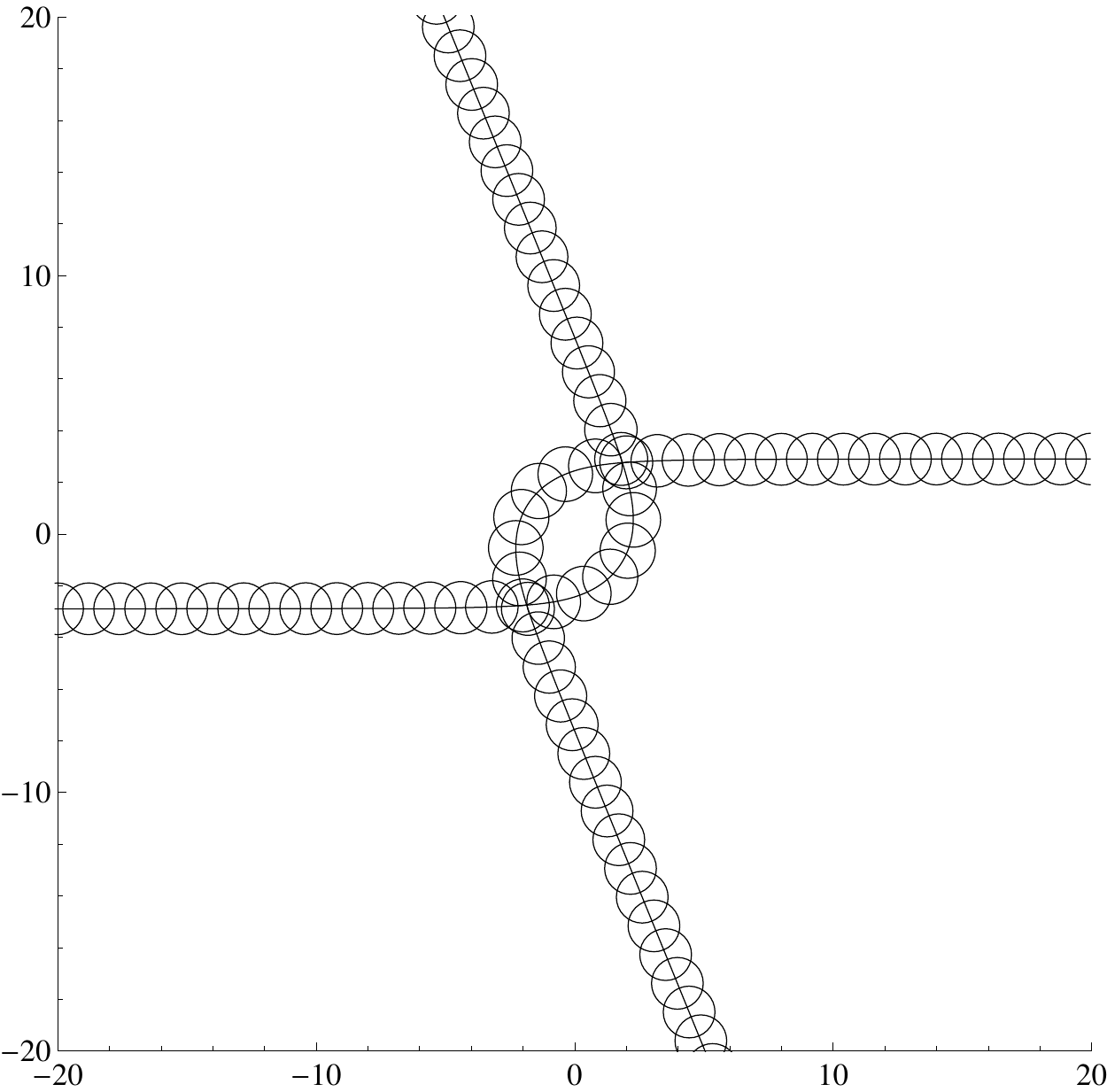}
\caption{A collision of two instantons with an impact parameter of $b = 2.9$.}
\label{fig:positive impact parameter scattering}
\end{figure}

\begin{figure}
\centering
\includegraphics[width=0.9\textwidth]{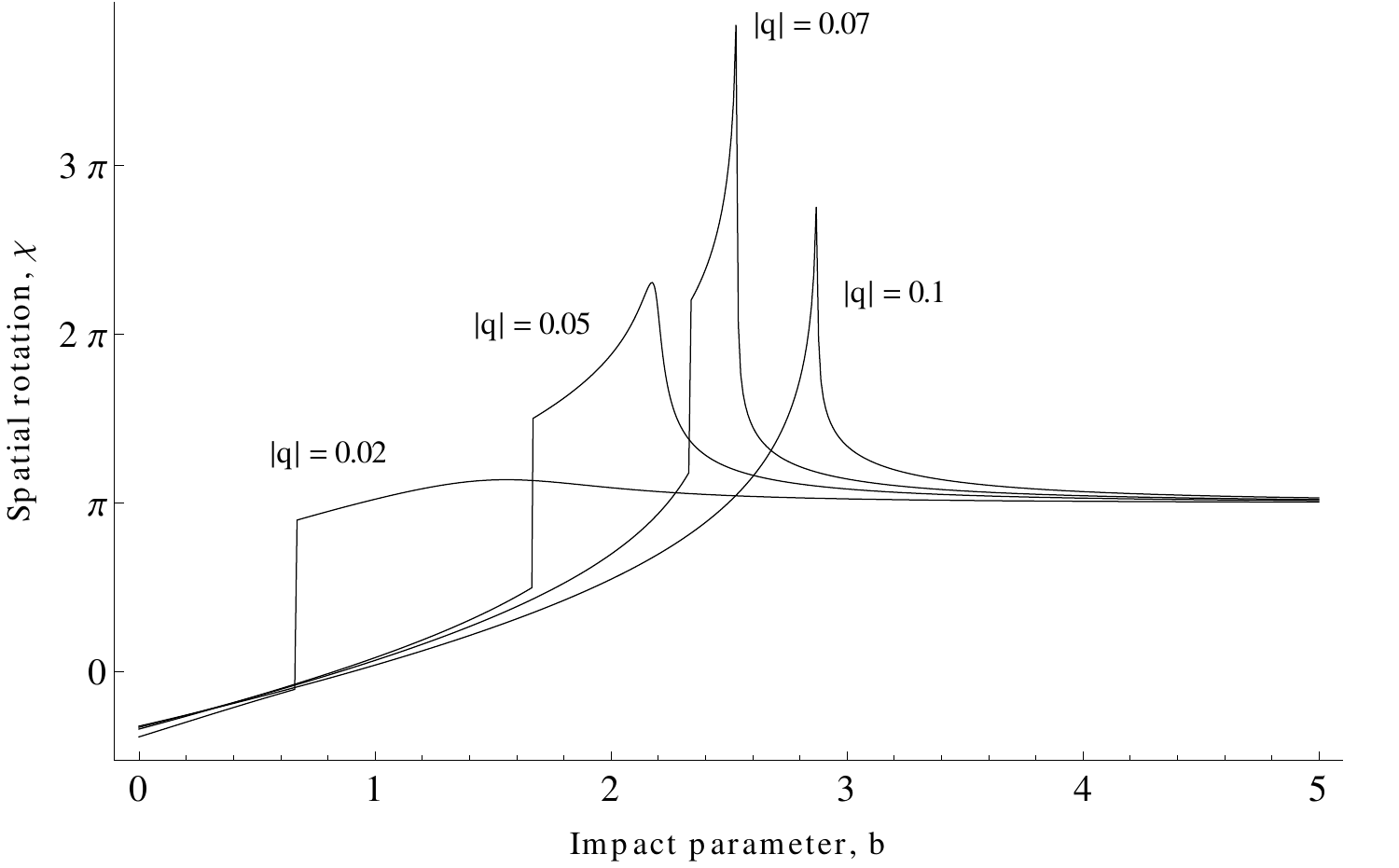}
\caption{The variation of the outgoing angle, $\chi$, with impact parameter for a collision between two dyonic instantons, shown for different values of the potential scale, $|q|$. The impact parameter, $b$, is in the positive direction.}
\label{fig:scattering angle vs positive impact parameter for dyonic instantons}
\end{figure}

Figure \ref{fig:scattering angle vs gauge angle} shows how the scattering angle depends on the initial difference in gauge angle between the two dyonic instantons. As with instantons, the scattering angle interpolates between zero when the gauge angles are equal and the value seen previously when the gauge angles are orthogonal. The strength of the interaction between the dyonic instantons again depends on the difference in their gauge angles with the strongest interaction occurring when they are orthogonal. At the other extreme when the gauge alignment is parallel the dyonic instantons are completely non-interacting.

\begin{figure}
\centering
\includegraphics[width=0.8\textwidth]{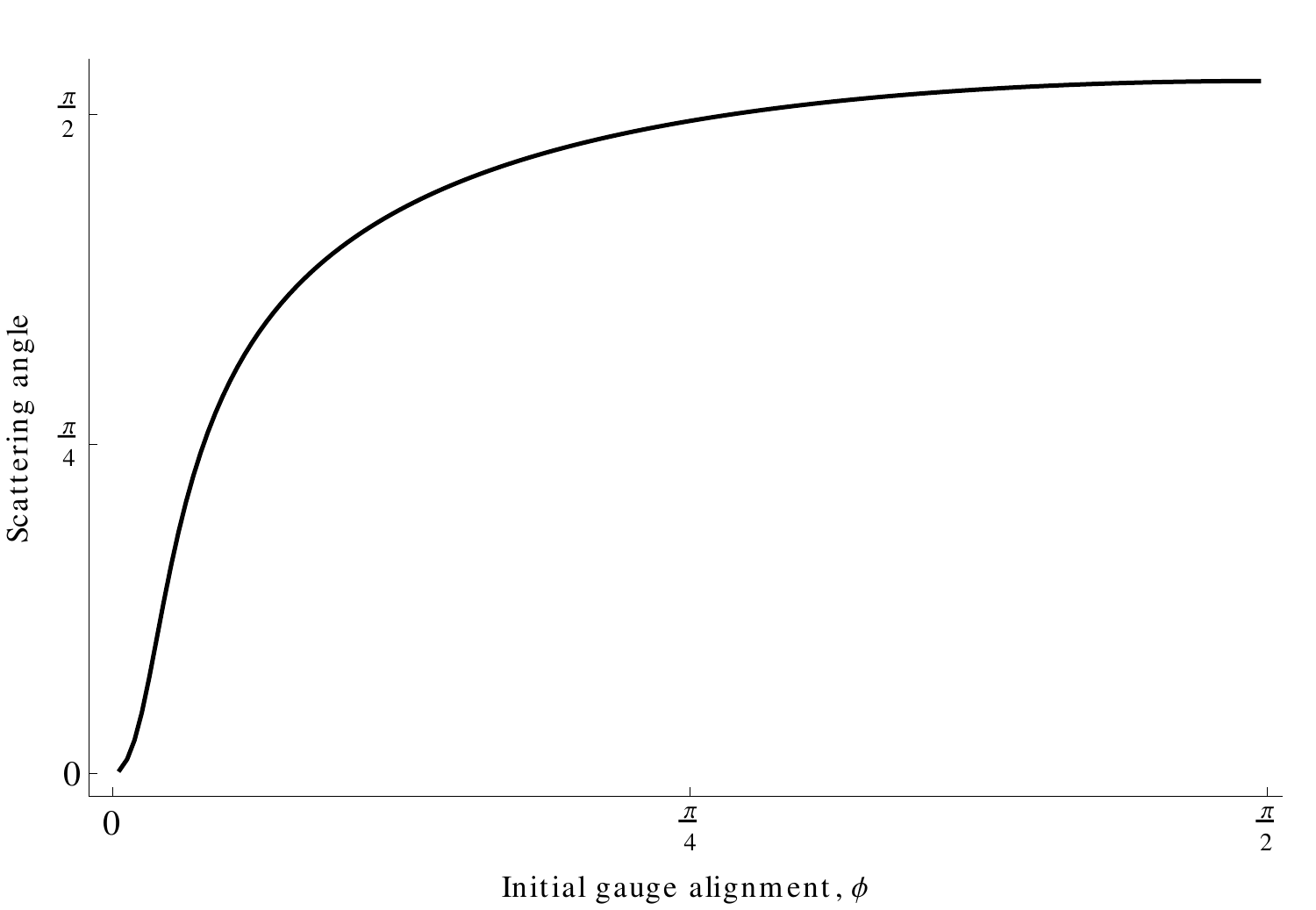}
\caption{The variation of the scattering angle with the initial gauge alignment, $\phi$ for two dyonic instantons with an impact parameter of $b = -0.5$.}
\label{fig:scattering angle vs gauge angle}
\end{figure}

The properties of dyonic instantons that we have considered so far are reminiscent of the Q-lumps considered by Leese \cite{Leese1991283}: both systems have a topological charge and a non-topological Noether charge; the presence of the non-topological charge induces a potential in the effective action for slow moving solitons and the potential stabilises the solitons against spreading out indefinitely under a small perturbation. Both systems also see similar scattering behaviour with head on collisions causing a deflection of more than $90^\circ$ before the lumps become coincident. As far as we are aware dyonic instantons and Q-lumps are the only solitons which have been seen to scatter in this way. Both system also have trajectories where the lumps orbit each other briefly when the impact parameter is in an appropriate range. Leese makes the point that it is difficult to avoid some external perturbation which would introduce a potential and so Q-lumps may be more physically relevant than the underlying pure $\sigma$-model soliton. This seems to be particularly relevant for instantons on D$4$-branes where the potential is induced by a separation of the branes.

\section{Geodesic completeness of the moduli space}

It is straightforward to see that the instanton moduli space is not geodesically complete, but the equivalent question for motion in the presence of a potential is not so straight forward. For pure instantons, a small negative perturbation in the size parameter will cause the instanton to shrink steadily until it hits the zero size singularity. For dyonic instantons however, there is a non-zero conserved angular momentum on the moduli space from the rotation in the unbroken $\U(1)$ gauge group. The angular momentum for a single (dyonic) instanton is simply
\begin{equation}
L = \rho^2 \dot \theta.
\end{equation}
This angular momentum protects a single dyonic instanton from shrinking to zero size under small perturbations.

For two dyonic instantons however the angular momentum is more complicated and the picture is not as clear.
On the two instanton moduli space the conserved gauge angular momentum arises from Killing direction $\Theta$ in the metric and is given by
\begin{equation}
L = g_{\Theta p} \dot z^p,
\end{equation}
where $\Theta$ is the embedding angle in the unbroken $\U(1)$ as in equation \eqref{eq:polar coordinate metric}. Considering just the complex subspace the angular momentum for two dyonic instantons is
\begin{equation}
L = \rho_1^2 \dot \theta_1 + \rho_2^2 \dot \theta_2 - \frac{2}{N_A} \rho_1 \rho_2 \cos \phi \sin \phi\, ( \rho_1 \dot \rho_2 - \rho_2 \dot \rho_1 ) - \frac{2}{N_A} \rho_1^2 \rho_2^2 \sin^2\phi\, (\dot \Theta - 2 \dot \chi).
\end{equation}
Since there is only an overall conserved quantity the individual instantons are free to transfer angular momentum. It is no longer clear \emph{a priori} whether one of the dyonic instanton can shrink to zero size by exchanging angular momentum with the other.

By numerically exploring motion on the moduli space we have been readily able to find trajectories where the instantons do indeed exchange angular momentum in such a way that one instanton shrinks to zero size. This is most easily observed when the instantons are far enough apart to be clearly distinct yet still within range of interaction. An illustrative example is shown in Figure \ref{fig:oscillation to zero size} where both instantons start with a non-zero angular momentum but one draws angular momentum from the other until it passes through zero size. Both instantons continue to oscillate at a steady rate and so long as the instanton reaches the lowest point of its oscillation at the same time as passing through zero angular momentum it will hit the zero size singularity. This requires fine tuning of one of the parameters which we were able to achieve to the limits of our numerical accuracy. This fine tuning suggests there is a subset of initial conditions of codimension one which will evolve to hit a zero size singularity.

\begin{figure}
  \centering
  \begin{subfigure}[b]{0.47\textwidth}
    \centering
    \includegraphics[width=\textwidth]{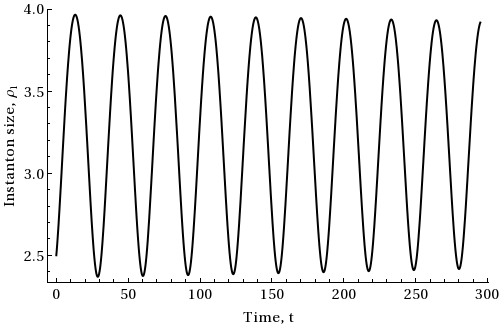}
    \caption{The size of the first instanton, $\rho_1$. This is drawing angular momentum away from the second.}
  \end{subfigure}
  \hfill
  \begin{subfigure}[b]{0.47\textwidth}
    \centering
    \includegraphics[width=\textwidth]{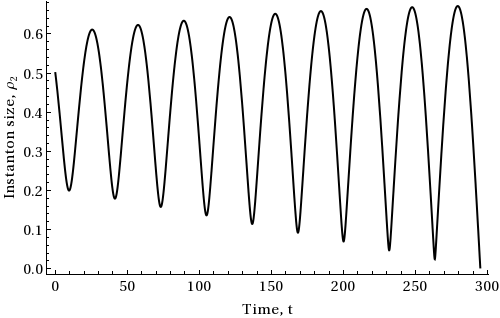}
    \caption{The size of the second instanton, $\rho_2$. This passes through zero size.}
  \end{subfigure}\\ \vspace{0.2in}
  \begin{subfigure}[b]{0.47\textwidth}
    \centering
    \includegraphics[width=\textwidth]{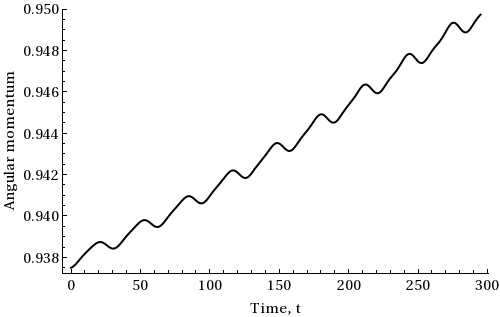}
    \caption{The angular momentum of the first instanton.}
  \end{subfigure}
  \hfill
  \begin{subfigure}[b]{0.47\textwidth}
    \centering
    \includegraphics[width=\textwidth]{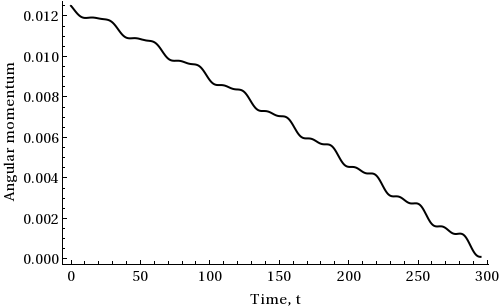}
    \caption{The angular momentum of the second instanton. This passes through zero as the size, $\rho_2$, reaches the lowest point of its oscillation.}
  \end{subfigure}
  \caption{The evolution to a zero size singularity of two initially non-singular dyonic instantons. The magnitude of the scalar VEV was $|q| = 0.1$ and the initial values were $\rho_1 = 2.5, \rho_2 = 0.5, \omega = 15, \phi = -\frac{\pi}{10}$. The initial velocities were $\dot \rho_1 = 0.1, \dot \Theta = 0.2, \dot \phi = 0.1$ and $\dot \rho_2 = -0.03$. All other initial velocities were zero.}
  \label{fig:oscillation to zero size}
\end{figure}

If we consider the full moduli space rather than just motion on the complex geodesic submanifold then we observe that the same generic behaviour is possible. The initial parameters now need a further two parameters to be fine tuned so that the additional two components of $v_1$ or $v_2$ also pass through their minimum value as the angular momentum passes through zero.

\section{Localised charge two instantons}

In this section we will consider the charge two object formed by two coincident (dyonic) instantons. For pure instantons, this configuration cannot be considered as an individual object as a small perturbation to the instanton positions will cause the two constituent instantons to drift apart until they are well separated again. Dyonic instantons however, are stabilised at a fixed separation by the potential. Figure \ref{fig:orbiting dyonic instantons} shows the result of giving two separated dyonic instantons a small initial velocity away from each other. The dyonic instantons now orbit each other in a spiralling pattern. 

\begin{figure}
\centering
\includegraphics[width=0.8\textwidth]{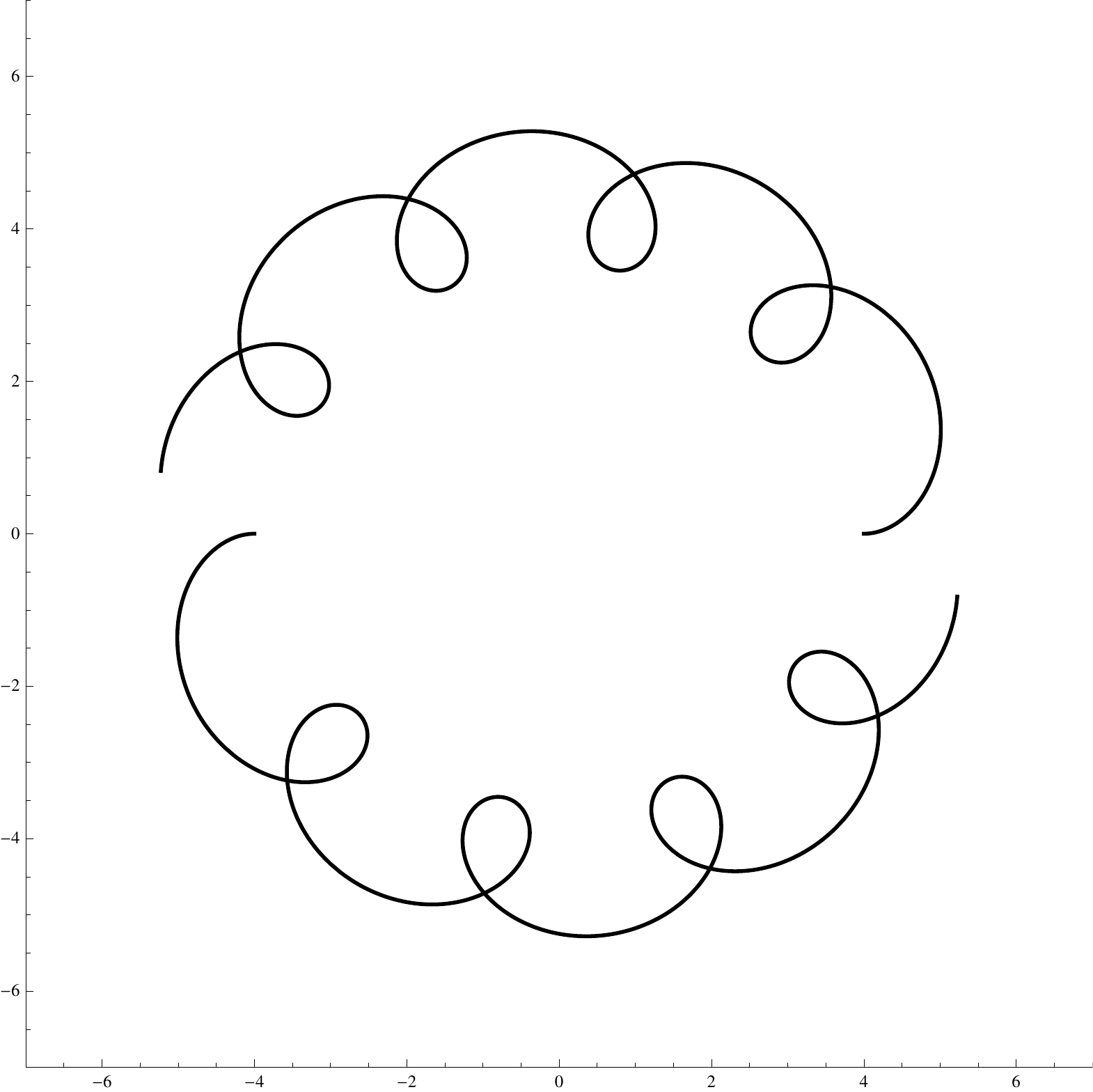}
\caption{Two dyonic instantons in a stable orbit after a small outwards perturbation in their positions. The instantons started at a separation of $x = 4$ with size $\rho = 1$ and a potential scale of $q = 0.1$. They were given an initial outwards velocity of $v = 0.0005$.}
\label{fig:orbiting dyonic instantons}
\end{figure}

The dyonic instantons will only form a stable orbit for a small enough perturbation and will otherwise continue to move away from each other at a steady speed. When only moving slightly faster than this `escape velocity' the dyonic instantons display some orbiting behaviour but do not settle into a stable orbit. Figure \ref{fig:max orbit velocity vs separation} shows how the separation affects the threshold velocity at which the dyonic instantons will no longer form a stable orbit. The velocity decreases as the strength of the interaction between the lumps decreases. The maximum threshold velocity is located close to where the instantons are coincident but with a slight shift towards a finite separation.

\begin{figure}
\centering
\includegraphics[width=0.8\textwidth]{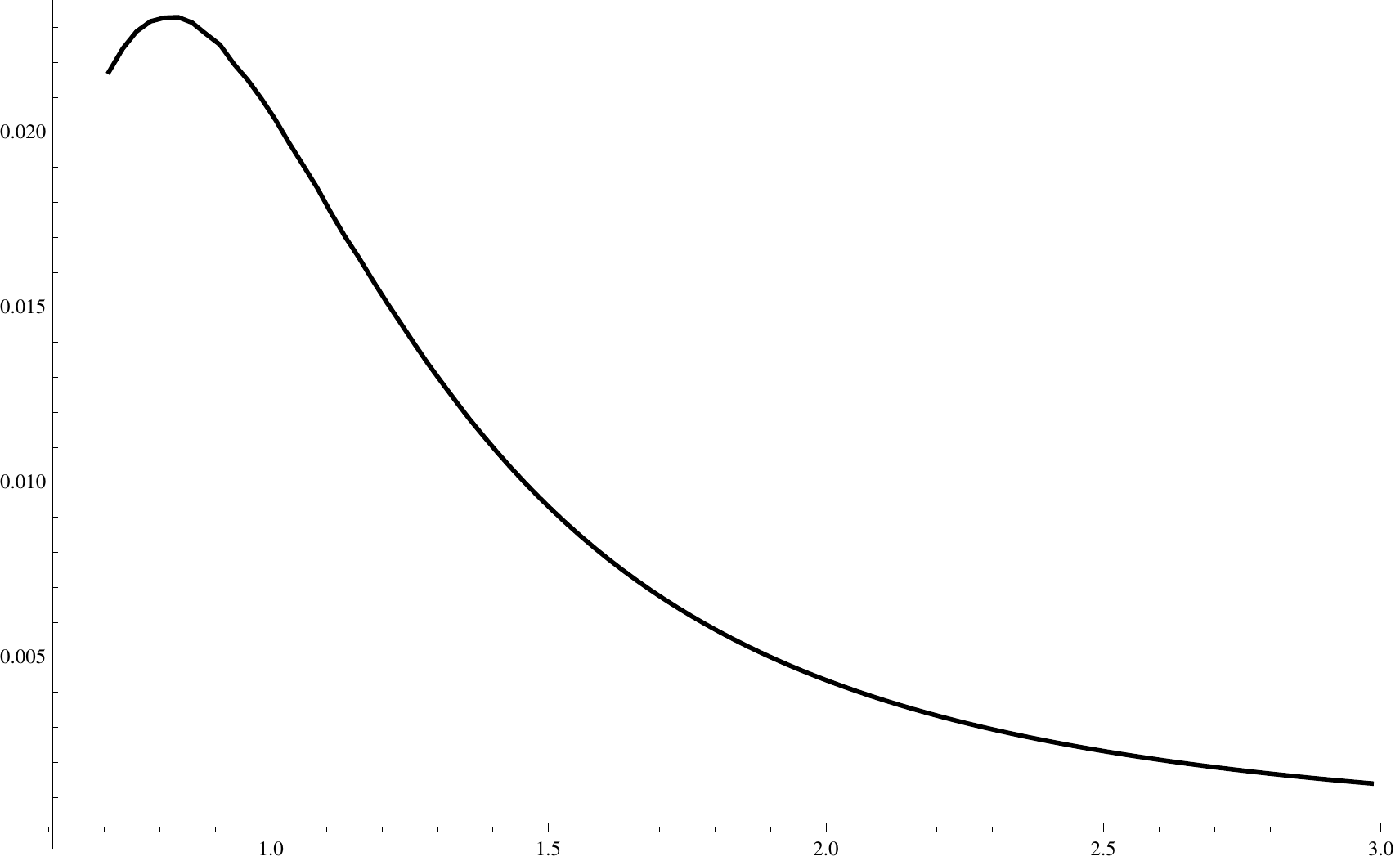}
\caption{The maximum outwards velocity of two dyonic instantons that will lead to a stable orbit with different separations. The instantons have size $\rho = 1$ and a potential scale of $|q| = 0.1$.}
\label{fig:max orbit velocity vs separation}
\end{figure}

As a result of this stability, the charge two lump corresponding to two individual dyonic instantons being coincident is a stable object. It will not separate into two distinct charge one dyonic instantons under a perturbation.

The charge two dyonic instanton also admits the only closed geodesic that we have found. On this geodesic the instantons remain exactly coincident,
\begin{equation}
\omega = \frac{\rho}{\sqrt{2}}, \quad \dot \omega = \ddot \omega = 0,
\end{equation}
and have no oscillations in their size,
\begin{equation}
\dot \rho = \ddot \rho = 0.
\end{equation}
The equations of motion are then satisfied by
\begin{equation}
\dot \chi = - 4v, \quad \dot \theta = v.
\end{equation}
This corresponds to a spatially rotating charge two instanton with the rotational velocity $\dot\chi$ fixed by the scale of the potential. 

It would be interesting to investigate whether such a closed geodesic is stable in the full field theory or whether there are higher order radiative corrections that would cause it to decay. We leave this for future consideration.

\section{Conclusions and outlook}

In this paper we have calculated the full metric and potential on the moduli space of two dyonic instantons in terms of the parameters in the ADHM construction. With this construction in mind we have been able to understand some of the structure of the moduli space as arising from the quotient of the moduli space by symmetries of the ADHM data. We have also explored the dynamics of two slow moving (dyonic) instantons using the moduli space approximation. We have seen that instantons readily undergo right angled scattering like many other soliton systems. This too can be understood from symmetries of the underlying ADHM data. The presence of a potential has a significant effect on the motion of dyonic instantons and we have seen that they behave in a way which is qualitatively similar to Q-lumps \cite{Leese1991283}.

Several questions remain open for future research. We have only explored the dynamics when the (dyonic) instantons lie in a plane with their relative gauge alignment in the unbroken $\U(1)$ symmetry. When the gauge alignment in the full $\SU(2)$ symmetry is orthogonal to the unbroken $\U(1)$, the final term in the potential vanishes. It would be interesting to explore the effects of this on the dynamics of collisions between two dyonic instantons. Unfortunately the complexity of the full moduli space makes a systematic exploration of this regime numerically expensive.

In our discussion of the dynamics we have assumed that the moduli space approximation is a suitable approximation in the regimes we have considered and we have discounted any radiative modes as negligible. Certainly this is the case in similar systems \cite{Leese1991283} and we expect it to hold here as well, but ideally we could check the validity of the approximation with a comparison to the full field theory. Unfortunately a full simulation of the four dimensional field theory is beyond the reach of available computing power at this time. It may be possible to revisit this question in the future.

So far we have only considered the classical behaviour of this system in the context of the world volume theory of $2$ D$4$-branes. It would be interesting to explore how these results relate to M$5$-brane theory where the instanton moduli space must be recovered as part of the compatification limit. For instance, the dyonic instanton moduli space can be explicitly recovered in the light-cone quantisation of the equations of motion of a non-abelian system with $(2,0)$ supersymmetry \cite{Lambert:2010wm, Lambert:2011gb}. This system is not a full description of the M$5$-brane system but still captures some of the relavant properties.

To fully understand the connection between $5$ dimensional super-Yang-Mills and the world volume theory of multiple M$5$-branes it will be necessary to understand how the theory behaves near to the singularities on the moduli space, particularly in the quantised theory. One possibility for regulating the singularities is via a non-commutative deformation of the instantons. This would place a minimum bound on the instantons' size \cite{Nekrasov:1998ss}. Work on calculating the moduli space of non-commutative instantons is in progress. Another possibility is that the wave functions are invariant under the symmetries which are responsible for the orbifold structure of the moduli space and therefore unaware of the singularities.

In the quantum theory it would be interesting to investigate the bound states of dyonic instantons. Previous studies in this direction (for example \cite{Lee:1999xb}) have considered the bound states of periodic instantons and the behaviour in the decompactification limit. With the full metric and potential for two instantons now available it should be possible to consider this decompactified limit for two instantons directly. This may give some insights into the unusual $N^3$ scaling of the degrees of freedom expected for multiple M$5$-branes.

Finally, this work could be extended to calculate the moduli space metric and potential of dyonic instantons in $\SU(3)$ or $\SU(N)$. The higher gauge group allows the possibility of bound states that pass through another D-brane and may provide a more direct description of the index counting in \cite{Kim:2011mv}. Work on this is also in progress.

\section*{Acknowledgements}

We would like to thank Paul Sutcliffe, Kyung Kiu Kim, Seok Kim and Rafael Maldonado for useful discussions. We are grateful to the Center for Quantum Spacetime at Sogang University for their hospitality while part of this work was completed. DJS would also like to thank the Isaac Newton Institute for Mathematical Sciences, as part of the programme on the Mathematics and Applications of Branes in String and M-theory, and the Korea Institute for Advanced Study, for hospitality. JPA is supported by an STFC studentship and DJS is supported in part by the STFC Consolidated Grant ST/J000426/1. We also acknowledge the EU FP7-People-2009-IRSES grant ``Integrability, Symmetry And Quantum Spacetime'' (ISAQS).

\appendix

\section{Calculations}

\subsection{The scalar field} \label{ap:calculation of phi}

The scalar field can be calculated from the ADHM data without having to solve its equation of motion directly. We make use of the ansatz and method presented in \cite{Dorey:1996hu}. We start with an ansatz for $\phi$,
\begin{equation} \label{eq:phi ansatz in appendix}
\phi = i U^\dagger \mathcal{A} U, \quad \mathcal{A} = \begin{pmatrix}
q & 0 \\
0 & P
\end{pmatrix},
\end{equation}
where $U$ is the null vector from the ADHM construction, $q$ is a pure imaginary quaternion and $P$ is a $k \times k$ real and anti-symmetric matrix. The scalar VEV is given by $iq \in \mathfrak{su}(2)$.

The matrix $P$ is to be determined by solving the equation of motion $D^2 \phi = 0$. A straight-forward but lengthy calculation gives 
\begin{equation}
D^2 \phi = - 4i U^\dagger \{ bfb^\dagger, \mathcal{A} \} U 
           + 4i U^\dagger b f \Tr_2( \Delta^\dagger \mathcal{A} \Delta) f b^\dagger U.
\end{equation}
The trace in the second term is only over the quaternionic blocks and therefore picks out twice the real part of its argument. Using the block diagonal form of $\mathcal{A}$ the first term is $-4 i U^\dagger \{f, P\} U$. The lower $k \times k$ block of the ADHM data is $\Omega' = \Omega - \id_2 x$ and so we can rewrite the trace in the second term as
\begin{align}
\Tr_2( \Delta^\dagger \mathcal{A} \Delta) &= \Tr_2 (\Lambda^\dagger q \Lambda) + \Tr_2 (\Omega'^\dagger P \Omega') \\
&= \Tr_2 (\Lambda^\dagger q \Lambda) 
   + \tfrac{1}{2} \Tr_2 ( [\Omega'^\dagger, P] \Omega' - \Omega'^\dagger [\Omega', P] + \{P, \Omega'^\dagger \Omega\} ) \\
&= \Tr_2 (\Lambda^\dagger q \Lambda) 
   + \tfrac{1}{2} \Tr_2 ( [\Omega'^\dagger, P] \Omega' - \Omega'^\dagger [\Omega', P] +\{P, f^{-1}\} - \{P, \Lambda^\dagger \Lambda \} ),
\end{align}
where we have used $\Delta^\dagger \Delta = \Lambda^\dagger \Lambda + \Omega'^\dagger \Omega' = f^{-1}$. In the commutator terms all $x$ dependence is proportional to $\id_k$ and vanishes. Thus
\begin{align}
\tfrac{1}{2} \Tr_2 ( [\Omega'^\dagger, P] \Omega' - \Omega'^\dagger [\Omega', P])
&= \tfrac{1}{2} \Tr_2 (\bar e_m e_n) ( [\Omega_m, P] \Omega_n - \Omega_m [\Omega_n, P]) \\
&= - [\Omega_m, [\Omega_m, P]].
\end{align}
Combining all of this we have
\begin{align}
\begin{split}
D^2\phi &= -4i \Big( U^\dagger \{ f,P - \tfrac{1}{2} \Tr_2(P)\} U \\
&\hspace{0.6in} + U^\dagger b f \left(\Tr_2 (\Lambda^\dagger q \Lambda) - [\Omega_m, [\Omega_m, P]] - \{P, \Lambda_m^\T \Lambda_m \} \right) f b^\dagger U \Big).
\end{split}
\end{align}
Since $P$ is real, the quantity $(P - \tfrac{1}{2} \Tr_2(P))$ is zero. Finally, we can rewrite
\begin{equation}
\Tr_2 (\Lambda^\dagger q \Lambda) = q^a \Lambda^\T_m \Lambda_n \Tr_2(\bar e_m e_a e_n) = 2 \eta^a_{mn} q^a \Lambda^\T_m \Lambda_n.
\end{equation}
So $P$ must satisfy
\begin{equation}
2 \eta^a_{mn} q^a \Lambda^\T_m \Lambda_n - [\Omega_m, [\Omega_m, P]] - \{P, \Lambda^\T_m \Lambda_m \} = 0.
\end{equation}
We can see from the symmetry properties of the other quantities involved that $P$ must be antisymmetric as expected.
Note that the indices in this expression are for the quaternion components, not the matrix components.

For instanton number $k=2$, the ADHM data is given in equation \eqref{eq:ADHM data}. The constituent parts are
\begin{equation}
\Lambda = \begin{pmatrix} v_1 & v_2 \end{pmatrix}, 
\quad \text{and} \quad
\Omega = \begin{pmatrix}
\tau & \sigma \\
\sigma & -\tau
\end{pmatrix}.
\end{equation} 
The first term in our constraint on $P$ is therefore
\begin{align}
2 \eta^a_{ij} q^a \Lambda^\T_i \Lambda_j &= \Tr_2 \del{ \Lambda^\dagger q \Lambda }, \\
&= \frac{1}{2} \begin{pmatrix}
0 & \Tr\del{ \bar v_1 q v_2 - \bar v_2 q v_1 } \\
\Tr\del{ \bar v_2 q v_1 - \bar v_1 q v_2} & 0
\end{pmatrix},
\end{align}
where $\bar v_2 q v_1 - \bar v_1 q v_2$ is real. If we write $P$ as
\begin{equation}
P = \begin{pmatrix}
0 & p \\
-p & 0
\end{pmatrix},
\end{equation}
then the second and third terms in the constraint are
\begin{equation}
[\Omega_m, [\Omega_m, P]] = 4 p \begin{pmatrix}
0 & |\tau|^2 + |\sigma|^2 \\
-(|\tau|^2 + |\sigma|^2) & 0
\end{pmatrix},
\end{equation}
and
\begin{equation}
\{P, \Lambda^\T_m \Lambda_m \} = p \begin{pmatrix}
0 & |v_1|^2 + |v_2|^2 \\
- ( |v_1|^2 + |v_2|^2 ) & 0
\end{pmatrix}.
\end{equation}
The off diagonal entry in $P$ is therefore given by
\begin{equation} \label{eq:solution for p}
p = \frac{1}{2 N_A} \Tr \del{ \bar v_1 q v_2 - \bar v_2 q v_1 },
\end{equation}
where
\begin{equation}
N_A = |v_1|^2 + |v_2|^2 + 4 \del{ |\tau|^2 + |\sigma|^2 }.
\end{equation}

We will not present the expanded expression for $\phi$ as it is quite complicated and doesn't provide any additional insight. However, the implicit form in equation \eqref{eq:phi ansatz in appendix} allows up to calculate the potential arising from $\phi$ and to easily evaluate and plot $\phi$.

\subsection{The moduli space metric} \label{ap:calculation of metric}

In this section we will present the calculation of the metric on the moduli space of two instantons.
We will first review the method of Osborn used to calculate the metric determinant \cite{Osborn:1981yf}. This was used by Peeters and Zamaklar to calculate the moduli space metric \cite{Peeters:2001np} of two instantons but only to order $|\tau|^{-2}$. We have extended this calculation to the full metric.

Recall that the metric on the moduli space of instantons is defined by
\begin{equation}
g_{rs}= \int \dif{}^4x \, \Tr \left( \delta_r A_i \delta_s A_i \right),
\end{equation}
where $r,s = 1, \ldots, 8k$ correspond to the coordinates on the moduli space and
\begin{equation}
\delta_r A_i = \partial_r A_i - D_i \epsilon_r,
\end{equation}
are the zero-modes of $A_i$ corresponding to variations along these coordinates. The zero-modes have their gauge transformation component removed by $D_i \epsilon_r$ and so are orthogonal to gauge transformations,
\begin{equation}
D_i(\delta_r A_i) = 0.
\end{equation}

In principle we could find the metric by finding an explicit expression for $A_i(\mathbf{z}; \mathbf{x})$, solving the gauge fixing condition for $\epsilon_r$ and taking the trace of each pair of zero-modes. In practice this approach is intractable. Fortunately we can use the ADHM construction to reduce this to an algebraic calculation which can be readily done for two $\SU(2)$ instantons.

Recall that if the ADHM data is given by $\Delta(x) = a - bx$ then the gauge field is constructed as
\begin{equation}
A_i = i U^\dagger \partial_i U, \quad \text{where} \quad \Delta^\dagger U = 0, \quad U^\dagger U = \id.
\end{equation}
Let us work at a single point on the moduli space, $\mathbf{z}_0$. The derivative of $A_i$ in one of the coordinate directions on the moduli space (parameters in the ADHM data) can be calculated in terms of the ADHM data,
\begin{equation} \label{eq:variation of A wrt Delta}
\partial_r A_i \big|_{z = z_0} = - i U^\dagger \partial_r \Delta f \bar e_i b^\dagger U + i U^\dagger b e_i f \partial_r \Delta^\dagger U + D_i(i U^\dagger \partial_r U).
\end{equation}
The last term is an explicit gauge transformation but the first two terms are not necessarily orthogonal to gauge transformations and may also contain an implicit gauge transformation. However, we have freedom in rewrite $\partial_r A_i$ with different parts of the gauge transformation made explicit. Recall that we can perform a transformation of the ADHM data of the form
\begin{equation}
\Delta \rightarrow Q \Delta R, \quad U \rightarrow Q U.
\end{equation}
Consider a transformation of this form in a region of $\mathbf{z}_0$ such that $Q(\mathbf{z}_0) = \id$ and $R(\mathbf{z}_0) = \id$. This leaves $A_i$ invariant and allows us to write
\begin{equation}
\partial_r A_i \big|_{z = z_0} = - i U^\dagger C_r f \bar e_i b^\dagger U + i U^\dagger b e_i f C_r^\dagger U + D_i(i U^\dagger \partial_r(Q^\dagger U)),
\end{equation}
where
\begin{equation}
C_r = \partial_r \Delta + \partial_r Q \Delta + \Delta \partial_r R.
\end{equation}
We can use this freedom to choose an appropriate $\partial_r Q$ and $\partial_r R$ at $\mathbf{z}_0$ such that the only piece of $\partial_r A_i$ which is parallel to a gauge transformation is in the explicit last term. The first two terms will then be a zero-mode. The conditions that $C_r$ must satisfy for this to be the case are expressed in the following claim.

\begin{claim}
The expression 
\begin{equation}
\delta_r A_i = - i U^\dagger C_r f \bar e_i b^\dagger U + i U^\dagger b e_i f C_r^\dagger U,
\end{equation}
is a zero-mode if $C_r$ is independent of $x$ and
\begin{equation}
\Delta^\dagger C_r = (\Delta^\dagger C_r)^\T.
\end{equation}
Equivalently, if
\begin{equation}
a^\dagger C_r = (a^\dagger C_r)^\T, \quad \text{and} \quad b^\dagger C_r = (b^\dagger C_r)^\T.
\end{equation}
\end{claim}
\begin{proof}
Consider the expression
\begin{equation}
a_i \equiv U^\dagger b f e_i,
\end{equation}
which makes up part of $\delta_r A_i$.
Treating this as a vector in the fundamental representation we can work out its covariant derivative,
\begin{align}
D_i a_j &\equiv \partial_i a_j - i A_i a_j \\
&= U^\dagger e_i b f \Delta^\dagger b f e_j + U^\dagger b f (\bar e_i b^\dagger \Delta + \Delta^\dagger b e_i) f e_j.
\end{align}
If we write $\Delta^\dagger b = c_k \bar e_k$ with the quaternion components made explicit and $c_k$ some real valued matrices then we can write this covariant derivative as
\begin{align}
D_i a_j &= U^\dagger b f c_k f \left( e_i \bar e_k e_j + \bar e_i e_k e_j + \bar e_k e_i e_j \right) \\
&= - U^\dagger b f c_k f \left( e_i \bar e_j e_k - 2 \delta_{jk} e_i - 2 \delta_{ik} e_j\right),
\end{align}
where we have used the quaternion identity,
\begin{equation}
\bar e_i e_j = - \bar e_j e_i + 2 \delta_{ij}.
\end{equation}
In this form is it easy to see that $a_i$ satisfies the linear self-dual field equation and background gauge condition,
\begin{equation}
D_{[i} a_{j]} = \tfrac{1}{2} \varepsilon_{ijkl} D_k a_l, \quad \text{and} \quad D_i a_i = 0.
\end{equation}
The covariant derivative of $\delta_r A_i$ can be written as
\begin{align}
D_i(\delta_r A_j) &= - i D_i U^\dagger C_r a_j^\dagger + i a_j C_r^\dagger D_i U - i U^\dagger C_r (D_i a_j)^\dagger + i D_i a_j C_r^\dagger U \\
&= - i U^\dagger b f \left(e_i \Delta^\dagger C_r \bar e_j - e_j C_r^\dagger \Delta \bar e_i \right) f b^\dagger U - i U^\dagger C_r D_i a_j^\dagger + i D_i a_j C_r^\dagger U.
\end{align}
Here $U^\dagger$ is also treated as a vector in the fundamental representation and its covariant derivative is
\begin{align}
D_i U^\dagger &\equiv \partial_i U^\dagger - i A_i U^\dagger \\
&= U^\dagger e_i b f \Delta^\dagger.
\end{align}
We have already shown that the last two terms in $D_i(\delta_r A_j)$ satisfy the conditions of a zero-mode so we only need to consider the first. Thus $\delta_r A_i$ will be a zero-mode if
\begin{equation}
K_{[ij]} = \tfrac{1}{2} \varepsilon_{ijkl} K_{kl}, \quad K_{ii} = 0
\end{equation}
where
\begin{equation}
K_{ij} = e_i \Delta^\dagger C_r \bar e_j - e_j C_r^\dagger \Delta \bar e_i.
\end{equation}
This is true if
\begin{equation}
\Delta^\dagger C_r = (\Delta^\dagger C_r)^\T.
\end{equation}
Since $\Delta = a - bx$ is linear in $x$ and $C_r$ has no dependence on $x$, we can split this into two conditions,
\begin{equation}
a^\dagger C_r = (a^\dagger C_r)^\T, \quad \text{and} \quad b^\dagger C_r = (b^\dagger C_r)^\T. \qedhere
\end{equation}
\end{proof}

It now remains to establish the conditions on $\partial_r Q$ and $\partial_r R$ so that the conditions on $C_r$ are satisfied and $\delta_r A_i$ is a zero-mode. In our canonical choice for the ADHM data $b$ is given by
\begin{equation}
b = \begin{pmatrix}
0 & 0 \\
1 & 0 \\
0 & 1
\end{pmatrix},
\end{equation}
and the transformation parameter $Q$ takes the form
\begin{equation}
Q = \begin{pmatrix}
q & 0  \\
0 & R^{-1}
\end{pmatrix}.
\end{equation}
We can set $q = 1$ so that it doesn't contribute to the variation of $Q$, which can now be expressed entirely in terms of the variation of $R$,
\begin{equation}
\partial_r Q = - b \, \partial_r R \, b^\dagger.
\end{equation}
The linear coefficient of $x$ in $C_r$ is therefore zero,
\begin{equation}
\partial_r b + \partial_r Q \, b + b \, \partial_r R = 0,
\end{equation}
and $C_r$ is indeed independent of $x$,
\begin{equation}
C_r = \partial_r a + \partial_r Q \, a + a \, \partial_r R.
\end{equation}

It is straightforward to see that $C_r$ satisfies the first condition to be a zero-mode, $a^\dagger C_r = (a^\dagger C_r)^\T$, since $R$ is an orthogonal matrix and $\partial_r R^\T = - \partial_r R$. For the second condition, $b^\dagger C_r = (b^\dagger C_r)^\T$, we require
\begin{equation} \label{eq:condition on delta R}
a^\dagger \partial_r a - (a^\dagger \partial_r a)^\T - a^\dagger b \, \partial_r R \, b^\dagger a - b^\dagger a \, \partial_r R \, a^\dagger b + \mu^{-1} \partial_r R + \partial_r R \, \mu^{-1} = 0,
\end{equation}
where
\begin{equation}
a^\dagger a = \mu^{-1},
\end{equation}
is real and invertible. To find the zero-mode in the $r$ direction we need to solve this constraint for $\partial_r R$ at each point of the moduli space in terms of parameters appearing in $a$. This is now a purely algebraic constraint on zero-modes.

To find the inner product between two zero-modes we can use the identity of Osborn \cite{Osborn:1981yf},
\begin{equation}
\Tr \left ( \delta_r A_i \delta_s A_i \right) = - \tfrac{1}{2} \partial^2 \Tr \left( C_r^\dagger P C_s f + f C_r^\dagger C_s \right).
\end{equation}
The metric is then given by using Stokes' theorem to integrate over the boundary,
\begin{align}
g_{rs} &= 2 \pi^2 \Tr \left( C_r^\dagger P_\infty C_s + C_r^\dagger C_s \right) \\
&= 2\pi^2 \Tr \left( \partial_r a^\dagger (1 + P_\infty) \partial_s a - \left( a^\dagger \partial_r a - (a^\dagger \partial_r a)^\T \right) \partial_s R \right), \label{eq:metric in terms of general ADHM data}
\end{align}
where
\begin{equation}
P = \id - \Delta^\dagger f \Delta
\end{equation}
and $P_\infty = \lim_{|x| \rightarrow \infty} P$.

Having outlined the general method, let us now turn our attention to the metric for two instantons in $\SU(2)$ Yang-Mills. As in Section \ref{sec:moduli space of two dyonic instantons}, the ADHM data for two instantons is
\begin{equation}
\Delta(x) = \begin{pmatrix}
v_1 & v_2 \\
\tilde \rho + \tau & \sigma \\
\sigma & \tilde \rho - \tau
\end{pmatrix}
- x \begin{pmatrix}
0 & 0 \\
1 & 0 \\
0 & 1
\end{pmatrix}.
\end{equation}
The projector at infinity is
\begin{equation}
P_\infty = \lim_{x \rightarrow \infty} P = \id - b^\dagger b = \begin{pmatrix}
1 & 0 & 0 \\
0 & 0 & 0 \\
0 & 0 & 0
\end{pmatrix}.
\end{equation}
The first part of the metric expression in equation \eqref{eq:metric in terms of general ADHM data} is therefore
\begin{align}
\dif s^2_1 &= \Tr \left( \dif a^\dagger (1 + P_\infty) \dif a \right) \\
&=2 \Tr (\dif\bar\tilde\rho \dif \tilde\rho + \dif \bar v_1 \dif v_1 + \dif \bar v_2 \dif v_2 + \dif\bar\tau \dif\tau + \dif\bar\sigma \dif\sigma ).
\end{align}
The $\tilde \rho$ directions are flat and we will neglect them from now on. The first four terms are all fundamental parameters, but the last term involving $\sigma$ needs to be expanded according to equation \eqref{eq:sigma definition},
\begin{equation}
\begin{split}
\MoveEqLeft 2\Tr(\dif \bar \sigma \dif \sigma) \\
&=            \Tr \Bigg( \frac{1}{8 | \tau |^2} \dif \bar \Lambda \dif \Lambda 
             + \frac{1}{4 | \tau |^4} \bar\Lambda \dif\bar\tau \tau \dif\Lambda
            - \frac{1}{4 | \tau |^4} \dif\bar\Lambda \Lambda \dif |\tau|^2
             + \frac{1}{8 | \tau |^4} |\Lambda|^2 \dif\bar\tau \dif\tau \\
&\qquad - \frac{1}{4 | \tau |^6} |\Lambda|^2 \dif\bar\tau \tau \dif |\tau|^2
             + \frac{1}{8 | \tau |^6} |\Lambda|^2 \dif |\tau|^2 \dif |\tau|^2 \Bigg).
\end{split}
\end{equation}
We note that
\begin{equation}
\Tr(\dif|\tau|^2) = \Tr(\dif \bar \tau \tau + \bar \tau \dif \tau) = 2 \Tr(\dif \bar \tau \tau),
\end{equation}
so that the terms at order $|\tau|^{-6}$ all vanish.
The terms at order $|\tau|^{-2}$ have been calculated previously \cite{Peeters:2001np} and are
\begin{equation}
\begin{split}
\MoveEqLeft \frac{1}{8|\tau|^2} \Tr( \dif \bar\Lambda \dif\Lambda ) \\
&=\frac{1}{|\tau|^2} \Big( 
              |v_1|^2 (\dif v_2 \cdot \dif v_2 ) + |v_2|^2 (\dif v_1 \cdot \dif v_1 )
              + 2(v_1 \cdot \dif v_1) (v_2 \cdot \dif v_2) \\
&\qquad - (\dif v_2 \cdot v_1) ( \dif v_2 \cdot v_1 ) - (\dif v_1 \cdot v_2) ( \dif v_1 \cdot v_2 )
              - 2 (v_1 \cdot v_2) (\dif v_1 \cdot \dif v_2) \\
&\qquad + 2 \varepsilon_{ijkl} v_1^k v_2^l \dif v_1^m \dif v_2^n \Big),
\end{split}
\end{equation}
where
\begin{equation}
p \cdot q = p_a q^a
\end{equation}
is the scalar product of quaternions treated as four vectors.
The terms at order $|\tau|^{-4}$ are
\begin{equation}
\frac{1}{8 |\tau|^4} \Tr( |\Lambda|^2 \dif\bar\tau \dif\tau  ) = 
  \frac{1}{|\tau|^4} \left( |v_1|^2 |v_2|^2  - (v_1 \cdot v_2)^2 \right) (\dif\tau \cdot \dif \tau),
\end{equation}
and
\begin{align}
\MoveEqLeft \frac{1}{4 |\tau|^4} \Tr( \bar \Lambda \dif \bar \tau \tau \dif \Lambda - \dif \bar \Lambda \Lambda \dif\,(\bar \tau \tau) )  \\
&= - \frac{1}{4 |\tau|^4} \Tr( \dif \Lambda \bar \Lambda  \bar \tau \dif\tau )  \\
&= - \frac{1}{4 |\tau|^4} ( \Tr( \Re(\dif \Lambda \bar \Lambda) \Re(\bar \tau \dif\tau) ) - \Tr( \Im(\Lambda \dif\bar \Lambda) \Im(\bar \tau \dif\tau) ) \\
\begin{split}
&=  - \frac{2}{|\tau|^4} \left( \tau \cdot \dif \tau \right) \Big( 
     |v_1|^2 (v_2 \cdot \dif v_2) + |v_2|^2 (v_1 \cdot \dif v_1) \\
& \qquad \qquad \qquad \qquad - (v_1 \cdot v_2) (v_1 \cdot \dif v_2) - (v_1 \cdot v_2) (v_2 \cdot \dif v_1)
    \Big) \\
&\qquad + \frac{1}{2 |\tau|^4} \left( \varepsilon_{ijkl} \Lambda_i \dif \Lambda_j \tau_k \dif \tau_l + (\Lambda \cdot \dif \tau)(\tau \cdot \dif \Lambda) - (\Lambda \cdot \tau) (\dif \Lambda \cdot \dif \tau) \right).
\end{split}
\end{align}
In this last line we have used
\begin{align}
\Tr( \Im(p \bar q) \Im(r \bar s) ) &= \Tr( \bar \eta^a_{ij} \bar \eta^b_{kl} e_a e_b ) p_i q_j r_k s_l \\
&= 2 ( \varepsilon_{ijkl} - \delta_{ik} \delta_{jl} + \delta_{il} \delta_{jk} ) p_i q_j r_k s_l.
\end{align}

For the second part of the metric, recall that $R$ is an $\O(2)$ transformation with one parameter, $\theta$. Since we require a continuous transformation it must be a rotation and its variation is an anti-symmetric matrix,
\begin{equation}
\dif R = - \dif \theta \begin{pmatrix}
0 & 1 \\
-1 & 0
\end{pmatrix}.
\end{equation}
Let us define a shorthand quantity, $\dif k$, by
\begin{equation} \label{eq:definition of k}
a^\dagger \dif a - (a^\dagger \dif a)^\T \equiv \dif k \begin{pmatrix}
0 & 1 \\
-1 & 0
\end{pmatrix},
\end{equation}
where the matrix form is determined by the left hand side being real and anti-symmetric. The constraint placed on $\dif R$ by equation \eqref{eq:condition on delta R} becomes
\begin{equation}
\dif \theta = \frac{\dif k}{N_A},
\end{equation}
where
\begin{equation}
N_A = |v_1|^2 + |v_2|^2 + 4 \left( |\tau|^2 + |\sigma|^2 \right).
\end{equation}
The second part of the metric is therefore
\begin{align}
\dif s^2_2 &= - \Tr\left( \left( a^\dagger \dif a -(a^\dagger \dif a)^\T \right) \dif R \right) \\
&= - 4 \frac{\dif k^2}{N_A}.
\end{align}
Calculating $\dif k$ explicitly from equation \eqref{eq:definition of k}, we find
\begin{equation}
\dif k = \bar v_1 \dif v_2 - \bar v_2 \dif v_1 + 2 (\bar \tau \dif \sigma - \bar \sigma \dif \tau).
\end{equation}
This is necessarily real and can be checked from the definition of $\sigma$. This expression can be expanded as
\begin{align}
\dif k &= \tfrac{1}{2} \Tr ( \bar v_1 \dif v_2 - \bar v_2 \dif v_1 + 2 (\bar \tau \dif \sigma - \bar \sigma \dif \tau) ) \\
\begin{split}
&= \left( v_1 \cdot \dif v_2 \right) - \left( v_2 \cdot \dif v_1 \right) \\
&\qquad - \frac{2}{|\tau|^2} \left( \varepsilon_{mnpq} v_2^m v_1^n \tau^p \dif\tau^q + (v_2 \cdot \tau) (v_1 \cdot \dif\tau) - (v_1 \cdot \tau) (v_2 \cdot \dif\tau) \right)
\end{split}
\end{align}

Putting all of this together, the metric on the moduli space of two $SU(2)$ instantons is therefore,
\begin{equation}
\begin{split}
\frac{\dif s^2}{8 \pi^2}
&= \dif v_1^2 + \dif v_2^2 + \dif \tau^2 \\
&\qquad + \frac{1}{4|\tau|^2} \Big(|v_1|^2 \dif v_2^2 + |v_2|^2 \dif v_1^2 + 
    2 (v_1 \cdot \dif v_1) (v_2 \cdot \dif v_2) - (v_1 \cdot \dif v_2)^2 \\
&\hspace{0.6in} - (v_2 \cdot \dif v_1)^2 - 2 (v_1 \cdot v_2) (\dif v_1 \cdot \dif v_2) + 2 \varepsilon_{ijkl} v_1^i v_2^j \dif v_1^k \dif v_2^l \Big) \\
&\qquad + \frac{1}{4|\tau|^4} \left(|v_1|^2 |v_2|^2 - (v_1 \cdot v_2)^2 \right) \dif\tau^2 \\
&\qquad - \frac{1}{2|\tau|^4} \Big(|v_1|^2 (v_2 \cdot \dif v_2) + |v_2|^2 (v_1 \cdot \dif v_1)\\ 
&\hspace{0.6in} - (v_1 \cdot v_2)(v_1 \cdot \dif v_2) - (v_1 \cdot v_2)(v_2 \cdot \dif v_1) \Big) \tau \cdot \dif \tau \\
&\qquad +\frac{1}{8|\tau|^4} \Big( \varepsilon_{ijkl} \Lambda_i \dif \Lambda_j \tau_k \dif \tau_l +(\Lambda \cdot \dif \tau)(\tau \cdot \dif \Lambda) - (\Lambda \cdot \tau) (\dif \Lambda \cdot \dif \tau)  \Big) \tau \cdot \dif \tau \\
&\qquad - \frac{1}{N_A} \Big(v_1 \cdot \dif v_2 - v_2 \cdot \dif v_1 \\
&\hspace{0.6in} - \frac{2}{|\tau|^2} \big(
       \varepsilon_{mnpq} v_2^m v_1^n \tau^p \dif\tau^q + (v_2 \cdot \tau) (v_1 \cdot \dif\tau) - (v_1 \cdot \tau) (v_2 \cdot \dif\tau)
     \big)
    \Big)^2.
\end{split}
\end{equation}

\subsection{The moduli space potential} \label{ap:calculation of potential}

The potential on the moduli space can be calculated directly from the Yang-Mills potential,
\begin{equation}
V = \int \dif{}^4x \, \Tr \del{ D_i \phi D_i \phi },
\end{equation}
by finding $\phi$ and $D_i \phi$ in terms of the ADHM data and evaluating the integral. 
We have seen in Appendix \ref{ap:calculation of phi} that $\phi$ is given by
\begin{equation}
\phi = i U^\dagger \mathcal{A} U, \quad \mathcal{A} = \begin{pmatrix}
q & 0 & 0 \\
0 & 0 & p \\
0 & -p & 0
\end{pmatrix},
\end{equation}
where $iq$ is the VEV of $\phi$ and $p$ is
\begin{equation}
p = \frac{1}{2 N_A} \Tr \del{\bar v_1 q v_2 - \bar v_2 q v_1}.
\end{equation}
If we integrate the potential by parts and use the scalar equation of motion for $\phi$, the potential becomes
\begin{equation}
V = \lim_{R \rightarrow \infty} \int_{\mathclap{|x| = R}} \dif\mathrm{S}^3\, \hat x_i \Tr \del{ \phi D_i \phi }.
\end{equation}
In terms of the ADHM data, the covariant derivative of $\phi$ is
\begin{equation}
D_i \phi = i U^\dagger e_i b f \Delta^\dagger \mathcal{A} U + i U^\dagger \mathcal{A} \Delta f \bar e_i b^\dagger U.
\end{equation}
For two instantons with ADHM data as in equation \eqref{eq:ADHM data}, the components of $U$ must satisfy,
\begin{align}
\bar v_1 U_1 + (\bar \tau - \bar x) U_2 + \bar \sigma U_3 &= 0, \\
\bar v_2 U_1 + \bar \sigma U_2 - (\bar \tau + \bar x) U_3 &= 0,
\end{align}
which can be solved in the limit $|x| \rightarrow \infty$ by
\begin{align}
U_1 &\rightarrow 1 \\
U_2 &\rightarrow \frac{x}{|x|^2} \bar v_1, \\
U_3 &\rightarrow \frac{x}{|x|^2} \bar v_2.
\end{align}
Expanding the leading order terms in the potential gives us
\begin{equation}
\hat x_i D_i \phi = 2 \frac{i}{|x|^3} \left( v_2 p \bar v_1 - v_1 p \bar v_2 + q(|v_1|^2 + |v_2|^2) \right) + \mathcal{O}\left( |x|^{-4} \right),
\end{equation}
and the potential is
\begin{align}
V &= - 2 \lim_{R \rightarrow \infty} \int_{\mathclap{|x| = R}} \dif\mathrm{S}^3\, \frac{1}{|x|^3} \Tr \left( q (v_2 p \bar v_1 - v_1 p \bar v_2) + q^2 (|v_1|^2 + |v_2|^2) \right) + \mathcal{O}\left( |x|^{-4} \right) \\
&= - 4 \pi^2 \Tr \left( q^2 (|v_1|^2 + |v_2|^2) + (\bar v_1 q v_2 - \bar v_2 q v_1 ) p \right) \\
&= 8 \pi^2 |q|^2 \del{ |v_1|^2 + |v_2|^2 - \frac{1}{N_A}|\bar v_2 \hat q v_1 - \bar v_1 \hat q v_2|^2 }.
\end{align}

\bibliographystyle{unsrt}
\bibliography{bibliography}

\end{document}